\documentclass[11pt]{article}

\usepackage{cite}
\usepackage{amssymb, amsfonts, amsthm}
\usepackage[cmex10]{amsmath}
\usepackage{silence} 

\usepackage[font=footnotesize]{subcaption}

\usepackage{mathtools}                  
\usepackage{paralist}                   
\usepackage{multirow}
\usepackage[usenames,dvipsnames]{color} 
\usepackage[normalem]{ulem}             
\usepackage{cancel}

\usepackage{hyperref} 

\usepackage{graphicx}
\graphicspath{{./gfx/}}
\DeclareGraphicsExtensions{.eps}
\usepackage{epstopdf}
\usepackage{placeins} 

\newcommand{\One}{\mathchoice{\rm 1\mskip-4.2mu l}{\rm 1\mskip-4.2mu l}{\rm 1\mskip-4.6mu l}{\rm 1\mskip-5.2mu l}}

\newcommand{\ignore}[1]{}
\newcommand{\nolabel}[1]{}

\newcommand{\KLI}{I}

\renewcommand{\Pr}{{\sf P}}           
\newcommand{\Pro}{{\sf P}}            
\DeclareMathOperator{\EV}{{\sf E}}    
\DeclareMathOperator{\Exp}{{\sf E}}   
\DeclareMathOperator{\SE}{{\sf SE}}   

\newcommand{\Naturals}{\mathbb{N}}
\newcommand{\Reals}{\mathbb{R}}

\newcommand{\sset}[1]{\{#1\}}                           
\newcommand{\cset}[3][\middle |]{\left\{ {#2} \, #1 \, {#3} \right\}} 
\newcommand{\stopset}[2]{\inf \left\{ {#1} \, : \, {#2} \right\}} 

\newcommand{\calo}{{\mathcal{O}}}
\newcommand{\cFt}{{\mathcal{F}}_{t}}

\newcommand{\bR}{\mathbb{R}}
\newcommand{\bN}{\mathbb{N}}

\newcommand{\cA}{\mathcal{A}}
\newcommand{\cB}{\mathcal{B}}
\newcommand{\cC}{\mathcal{C}}

\newcommand{\cN}{\mathcal{N}}

\newcommand{\cJ}{\mathcal{J}}
\newcommand{\cR}{\mathcal{R}}

\newcommand{\cP}{\mathcal{P}}

\newcommand{\cF}{\mathcal{F}}

\newcommand{\ccg}{\cC_{\gamma}}

\newtheorem{theorem}{Theorem}
\newtheorem{definition}{Definition}
\newtheorem{proposition}{Proposition}

\newtheorem{lemma}{Lemma}

\DeclareMathOperator*{\esssup}{ess\,sup}

\hypersetup{
    colorlinks=true,%
    bookmarksnumbered=true,%
    bookmarksopen=true,%
    citecolor=blue,%
    urlcolor=blue,%
    unicode=true,           
    breaklinks=true         
}

\usepackage[norefs,nocites]{refcheck} 
\newcommand{\warning}[1]{}

\begin{document}

\title{ Second-Order Asymptotic Optimality  in  \\
Multisensor  Sequential Change Detection}

\author{%
Georgios~Fellouris%
\thanks{G.\ Fellouris is with the  Department of Statistics, University of Illinois, Urbana-Champaign, IL, USA, e-mail: fellouri@illinois.edu.}%
~~and~Grigory~Sokolov%
\thanks{G.\ Sokolov is with the Department of Mathematical Sciences, Binghamton University, Binghamton, NY, USA, e-mail: gsokolov@binghamton.edu.}%
}

\date{}
\maketitle

\begin{abstract} 
A  generalized multisensor sequential change detection problem is considered, in which  a  number of (possibly correlated) sensors monitor an environment
in real time,  the joint distribution of their observations is determined by a global  parameter vector, and  at some unknown time there is a change in an unknown subset of components of this  parameter vector. In this setup, we consider the problem of detecting the time of the change as soon as possible, while controlling the rate of false alarms. We establish  the second-order asymptotic optimality (with respect to   Lorden's criterion)  of various generalizations of the CUSUM rule;  that is,  we show that  their additional expected worst-case detection delay  (relative to the one that could be achieved if the affected subset was known) remains bounded as the rate of false alarm goes to 0,  for any possible subset of affected components.  This general framework incorporates the traditional multisensor setup in which only an unknown subset of sensors is affected by the change. The latter problem has a special structure which we exploit in order to obtain feasible representations of the proposed  schemes.  We  present the results of a simulation study where we  compare the proposed schemes with scalable detection rules that are only first-order asymptotically optimal.
Finally, in the special case that the change affects exactly one sensor,  we consider the scheme that runs in parallel the local CUSUM rules and  study the problem of specifying the local thresholds.
\end{abstract}

\section{Introduction}

Suppose that a decision maker sequentially  collects    data from multiple streams, coming for example from  sensors monitoring an environment.  At some unknown point in time there is an abrupt change in the system that is perceived by either all or only a  subset  of the deployed sensors. The problem then  is to combine the  sequentially acquired observations from all streams in order to detect the change as quickly as possible, while controlling the rate of false alarms below an acceptable level.   
In what follows, we will refer to the various streams as ``sensors'', although this general framework can be applied to various setups beyond environmental monitoring, such as intruder detection in computer security  \cite{blazek,blazek2,levy}, epidemic detection in bio-surveillance \cite{bock,burkom}.

If  the regime before and after the change is completely specified,
 one may  apply a classical sequential  change detection  algorithm, such as the Cumulative Sum (CUSUM) rule, which was proposed by Page \cite{page}  and was shown by  Moustakides \cite{moust1} to be  the best possible rule in  Lorden's \cite{lorden1} minimax setup  in the case of iid observations before and after the change. However,  in many applications that motivate this multisensor sequential change detection problem it is reasonable to assume that there is only partial information regarding the post-change regime. A typical assumption in the literature is that the change affects the distribution of only a subset of sensors whose identity is unknown.  In this context, one needs to design detection rules that account for this uncertainty,  and also to quantify how much is lost relative to the ideal case that the actual affected subset is known.

The simplest  multisensor  detection rule is the \textit{multichart} CUSUM, according to which each sensor runs locally the  CUSUM algorithm  and the fusion center stops  the first time an alarm is raised by  any sensor.  This is a  \textit{one-shot} scheme suitable for decentralized implementation, as each possibly remote sensor needs to communicate with the fusion center at most once, an important  advantage in sensor networks that are characterized by limited bandwidth and energy. Tartakovksy et al.\ \cite{blazek, blazek2}  showed that  this rule has  a strong asymptotic optimality property  when  the change affects  exactly \textit{one} sensor (see also Section 9.2 in \cite{TNB_book2014}).  However, this rule  is very inefficient when a large, or even a moderate, number of sensors are actually affected. This calls for detection rules that are robust, in the sense that they should perform well for a large class of scenarios regarding the subset of affected sensors.

In order to address this problem,  Mei  \cite{mei_bio}  proposed   raising an alarm when the  sum of the local CUSUM statistics exceeds a threshold, a rule to which we   will refer  as  SUM-CUSUM. Under the assumptions of  iid observations before/after the change, completely specified but arbitrary  pre/post-change distributions  (up to  an integrability condition) and independent sensors it was shown that  SUM-CUSUM is   asymptotically optimal for every possible  subset of affected sensors.  Moreover, based on simulation experiments  it was argued that SUM-CUSUM  performs better than the multichart CUSUM unless a very small number of sensors is actually affected. This scheme was further refined by Mei in \cite{mei_sympo}, where it was suggested to use the sum of only the top $L$ local CUSUM statistics
when  at most $L$ sensors can be affected. Liu et al.\ \cite{mei_techno} and Banerjee and Veeravalli \cite{tap} considered further modifications that  allow for only a subset of sensors to be employed at any given time.

A  different, CUSUM-based, multisensor detection scheme  was proposed by
Xie and Siegmund \cite{xie}. Specifically, the proposed rule in this work was a modification of the CUSUM rule that  corresponds to the case that  the change occurs with probability $\pi$ in each sensor, where $\pi$ is the actual proportion of affected sensors.  Asymptotic approximations  for the operating characteristics
of  this  detection rule  were obtained in the case of independent sensors, each of which observes  Gaussian iid observations whose mean may change from $0$ to some unknown value. Moreover, it was argued based on simulation experiments that this detection rule is more efficient than SUM-CUSUM and this phenomenon was explained on the basis that SUM-CUSUM does not incorporate  the fact that all sensors are assumed to have the same change-point.

The main motivation for our  work was to propose and study  multisensor detection rules that enjoy a stronger asymptotic optimality property than SUM-CUSUM and are more efficient in practice. However,  we consider a more general  setup that  incorporates applications that are not captured by the  traditional multisensor framework, such as the power outage detection problem treated by Chen et al.\ \cite{chenban}. Specifically, we assume that there is a global parameter vector that determines the joint distribution of the (possibly correlated)  data streams and that the  change modifies  a subset of the components of this parameter vector. The affected subset of components
is unknown, but it is assumed to belong to a given class of subsets, $\cP$, that reflects our prior information regarding the location of the change. For example,  in the extreme case that there is no information regarding the affected subset,  $\cP$ consists of all non-empty subsets of components.

In this generalized multisensor setup, we consider first of all the CUSUM-type rule that employs the Generalized  Likelihood Ratio (GLR) statistic with respect to the unknown subset, to which we refer as GLR-CUSUM.  Moreover, we consider two alternative,  \emph{mixture}-based CUSUM-type schemes. We show that all these  detection rules enjoy a  \textit{second-order} asymptotic optimality property, i.e., that their additional worst-case detection delay (relative to the one  that could be achieved if the affected subset was known) remains bounded as the rate of  false alarms goes to 0  for any possible subset of affected components. This should be contrasted with the usual, weaker asymptotic optimality property,    according to which the asymptotic  relative efficiency  goes to 1,  and to which we refer as \textit{first-order} asymptotic optimality.

Our general framework includes as a special case the traditional setup in which  the change affects only an unknown subset of sensors, leaving the  remaining ones completely unaffected. In this special case, assuming further independent sensors and an upper bound on the size of the affected subset, we  obtain  representations of the  GLR-CUSUM (and one of the proposed mixture-based rule)  which  are  feasible when the number of sensors is large.   Moreover, we  show that while the  detection rule that was proposed by Xie and Siegmund \cite{xie}  is also second-order asymptotically optimal for any choice of the parameter $\pi$, this parameter cannot be interpreted  in this scheme as the proportion of affected sensors. Thus, we provide a theoretical explanation for the empirically observed  superiority of this rule over  SUM-CUSUM, which is only first-order asymptotically  optimal, as well as its robustness with respect to  $\pi$.

In the  special case that it is known in advance that the change affects exactly \textit{one} sensor,   GLR-CUSUM reduces to the multichart CUSUM, according to which an alarm is raised when a  local CUSUM statistic exceeds a local threshold.  The selection of identical thresholds seems to be a default choice in the literature, although non-identical thresholds have also been considered (see, e.g.  \cite[p.467]{TNB_book2014}). We address the problem of threshold selection in this setup and we show that seemingly reasonable threshold selections may lead to non-robust behavior, as the multichart CUSUM may fail to be even   first-order asymptotically optimal. We then argue in favor of a  specific selection of thresholds that equalizes  the asymptotic \textit{relative} performance loss  under the various scenarios.

Let us discuss at this point the connection of our work with the literature of sequential change detection.  Lorden \cite{lorden1} assumed (in  a ``single-sensor'' setup)  that the  distribution of the observations belongs to an exponential family and is specified after the change up to an unknown parameter. In this context,  he established the first-order asymptotic optimality of a  CUSUM-type rule that is based on  (a truncated away from 0 version of) the  GLR statistic.   In the special case of detecting a shift in the mean of  Gaussian observations, Siegmund  and Venkatraman \cite{venka} obtained approximations for the operating characteristics of the corresponding GLR-CUSUM. Our work differs from these two papers  in that  in our setup there is  a finite number of scenarios in the post-change regime.
Nevertheless, we also use the term  GLR-CUSUM to describe the CUSUM-type rule whose detection statistic  is based on  maximizing  with respect to the affected subset, which is  the unknown parameter in our setup.

We should also  mention an alternative formulation of the multisensor sequential change detection problem,  in which the change-points in the various sensors are in generally different and the  goal is to detect the first of them. This problem was considered from a Bayesian point of view by Ragavan and Veeravalli \cite{rag} and Ludkovski \cite{lud},  where a model for the change propagation was assumed. A minimax approach was considered by Hadjiliadis et al.\ \cite{oly} and Zhang et al.\ \cite{oly2}, where it was shown that the multichart CUSUM  has a strong asymptotic optimality property with respect to an extended Lorden criterion in which the worst case is taken with respect to all individual change-points and the history of observations up to the minimum of the change-points.

Another closely related problem to the one consider in this paper is the  joint change detection and isolation problem, in which one is interested in isolating the affected subset  upon detecting the change. While the GLR-CUSUM provides a natural estimator for this subset, a proper treatment of this problem  requires a different  formulation, in which the misclassification rate is  also controlled together with the  rate of false alarms. For more details on the joint sequential change detection and isolation problem we refer to Nikiforov \cite{niki1, niki2}, Oskiper and Poor \cite{oski}, Tartakovksy \cite{tarta_multi}.

The rest of the paper is organized as follows: in Section \ref{sec:general_problem} we formulate a general multisensor sequential change detection problem and  establish the second-order asymptotic optimality of various CUSUM-based detection rules.
In Section \ref{sec:special_problem} we restrict ourselves to the traditional setup of independent sensors where only a subset of them is affected by the change.
In Section \ref{sec:one_affected} we focus on the  case that the change  affects exactly one sensor and the design of the multichart CUSUM.
In Section \ref{sec:simulation} we present the findings of a simulation study that  illustrates our asymptotic results. We conclude in Section \ref{sec:conclusions}, whereas  we present the proofs of most results in \hyperref[appen]{Appendix}.
In what  follows we set $x^{+}=\max\{x,0\}$, we use $|\cdot|$ to denote set cardinality,  and we write $x \sim y$ when $x/y \rightarrow 1$, $x\simeq y$ when $x - y \rightarrow 0$.

\section{A general multisensor sequential  quickest detection problem} \label{sec:general_problem}

\subsection{Problem formulation} \nolabel{sec:formulation}

Let $K$ be the number of sensors (streams) and $X_t^k$  the observation in sensor $k$ at time $t$, where $1 \leq k \leq K$ and  $t \in \bN=\{1, 2 \ldots\}$. We assume that  the random vectors  $\{X_t \equiv (X_t^1, \ldots, X_t^{K})\}_{ t\in \bN}$  are  independent over time, but we allow observations from different sensors at the same time instant to be dependent. Moreover, we assume that the distribution of the sequence $\{X_t\}_{ t\in \bN}$  changes at some  unknown deterministic time  $\nu \in \bN \cup\{0\}$. To be more specific, let $f(\cdot |\theta)$ be a family of densities with respect to some dominating $\sigma$-finite measure, where  $\theta$ is a $N$-dimensional parameter vector and $N$ not necessarily equal to  $K$.
Then, we assume that $X_t \sim f(\cdot |\theta_t)$, where
$\theta_t:=(\theta_t^1, \ldots, \theta_t^{N})$ is given by
\begin{align} \label{eq:setup:param_change}
    \begin{split}
        \theta_t^k &= 0 \quad  \forall \;  1 \leq k \leq N,  \quad t \le \nu,  \\
        \theta_t^k &=
        \begin{cases}
            1 , & \quad k \in \cA, \\
            0 , &  \quad k \notin \cA,
        \end{cases}
        \qquad t > \nu ,
    \end{split}
\end{align}
and $\cA \subset \{1, \ldots, N\}$ is a subset of components of $\theta_t$. Note that the change may affect the observations in all sensors and that while the  pre-change distribution is completely specified,  the post-change  distribution is specified up to subset $\cA$ which is typically  unknown.
However,  we   will assume that the true affected subset  belongs to a given class $\cP$ of subsets of $\{1, \ldots, N\}$.  For example, if it is known that the change will affect  exactly (resp.\ at most) $L$ components of $\theta$, then   $\cP=\cP_{L}$ (resp.\ $\cP = \overline{\cP}_{L}$), where
\begin{align} \label{eq:class_arl}
\begin{split}
    \cP_{L} &:=\{\cA : |\cA|= L\}, \\
    \overline{\cP}_{L} &:= \{ \cA : 1 \leq |\cA| \leq L \},
\end{split}
\end{align}
$|\cA|$ being the cardinality of subset $\cA$ and $1 \leq L \leq N$. Note  that   $\overline{\cP}_{N}$ corresponds to the case of complete ignorance regarding the affected subset. Moreover,  any class $\cP$ is a subset of a class of the form $\overline{\cP}_{L}$ for some $1 \leq L \leq N$.  However, in what follows we do not make any assumption about $\cP$ unless this is explicitly stated.

Let   $\Pro _{\infty}$ and $\Exp_\infty$ the  probability measure and the expectation  when the change never occurs ($\nu =\infty$) and  $\Pro^{\cA}_{\nu}$ and $\Exp^{\cA}_{\nu}$ the
probability measure and the expectation when the change occurs at time $\nu$ only in subset $\cA$.  A sequential change detection rule is an $\{\cFt\}$-stopping time, $T$, at which the fusion center raises an alarm declaring that the change has occurred. Here,   $\cFt$ the $\sigma$-algebra generated by all observations up to time $t$, i.e., $\cFt := \sigma (X_s; \; 1 \leq s \leq t)$.
Following Lorden's approach \cite{lorden1} we measure the performance of a  detection rule $T$ \textit{when the change occurs in subset $\cA$}  with the  worst (with respect to the change point) conditional expected delay given the worst possible history of observations up to the change point:
\[
    \cJ_{\cA} [T] := \sup _{\nu \ge 0} \esssup \Exp^{\cA}_\nu [ (T - \nu)^{+}  \, | \, \mathcal{F} _\nu ].
\]
Moreover, we control the rate of false alarms below an acceptable, user-specified level $\gamma \geq 1$, thus,  we restrict ourselves to sequential detection rules in the class $\ccg=\{T: \, \EV_{\infty}[T] \geq \gamma\}$.
We will be interested in sequential detection rules that attain the optimal performance
$ \inf_{T \in \ccg} \cJ_{\cA} [T]$
for every $\cA \in \cP$ asymptotically as $\gamma \rightarrow \infty$.

\begin{definition} \nolabel{def:ao:first-order}
    Suppose that the detection rule  $T^{*}$ can be designed  so that
    $T^{*} \in \ccg$ for any given $\gamma \geq 1$.   We will say that  $T^{*}$ is  asymptotically optimal of
  \begin{enumerate}
  \item[(i)]  first-order,  if  for every  $\cA \in \cP$ we have as $\gamma \rightarrow \infty$
    \begin{align*} \nolabel{eq:ao:first-order}
        \cJ_{\cA} [T^{*}] \sim         \inf_{T \in \ccg} \cJ_{\cA} [T],
    \end{align*}
    \item[(ii)] second-order,  if  for every  $\cA \in \cP$ we have
    \begin{align*} \nolabel{eq:ao:second-order}
        \cJ_{\cA} [T^{*}] -  \inf_{T \in \ccg} \cJ_{\cA} [T] = \calo (1),
    \end{align*}
    where $\calo(1)$ is a bounded term as $\gamma \rightarrow \infty$.
    \end{enumerate}
\end{definition}

\noindent \textit{Remark:} If  $T^{*}$ is  first-order  asymptotically optimal, then it is possible that for some $\cA \in \cP$ we have
$
\cJ_{\cA} [T^{*}] -  \inf_{T \in \ccg} \cJ_{\cA} [T] \rightarrow \infty,
$
since both terms go to infinity as $\gamma \rightarrow \infty$. In such a case,
  the \textit{additional} (worst-case) expected number of observations that $T^{*}$  requires (relative to the corresponding optimal detection delay that would be achievable if the actual affected subset was known)  is unbounded  as the rate of false alarms goes to $0$.  This motivates the stronger notion of second-order asymptotic optimality.\\

It is the main goal of this work to propose detection rules that are second-order asymptotically optimal. In order to do so, we need to characterize the optimal performance that is attained by the oracle detection rule.

\subsection{The oracle rule}
For any given subset $\cA$, the solution to the constrained optimization problem
$
 \inf_{T \in \ccg} \cJ_{\cA} [T]
$
is well known. In order to describe it, let us set
\begin{align} \label{eq:likelihood_ratio}
    \Lambda_t^{\cA} &:=  \frac{d\Pro_0^{\cA}}{d\Pro_\infty} (\cFt)
    \quad \text{and}  \quad
        Z_t^{\cA} :=  \log  \Lambda_t^{\cA}
\end{align}
and note that  due to the assumption of  independence over time  we have
\begin{align*}
    \Lambda_t^{\cA} &= \Lambda_{t-1}^{\cA} \cdot  \exp\{\ell_t^\cA\}  , \quad \Lambda_0^{\cA} := 1, \\
    Z_t^{\cA} &= Z_{t-1}^{\cA}+  \ell_t^\cA  , \quad Z_0^{\cA} := 0,
\end{align*}
where   $\ell_t^\cA$ is the log-likelihood ratio
\begin{align} \label{eq:log-likelihood_ratio}
    \ell_t^\cA  &:= \log \frac{d\Pro_0^{\cA}}{d\Pro_\infty} (X_t) .
\end{align}
The CUSUM rule  for detecting a change in subset $\cA$ is
\begin{align} \label{eq:cusum}
    S^{\cA}_b & = \inf \left\{  t \in \bN : \tilde{Y}^{\cA}_{t} \geq b \right\},
\end{align}
where $b \in \bR$ is a deterministic threshold and  $\tilde{Y}^{\cA}$ is the so-called CUSUM statistic
\begin{align} \label{eq:cusum:stat_neg}
    \tilde{Y}^{\cA}_{t} := Z_t^{\cA} - \min_{0 \leq s < t} Z_{s}^{\cA}, \quad t \in \Naturals ,
\end{align}
which may be equivalently defined through the following recursion
\begin{align} \label{recursion2}
   \tilde{Y}^{\cA}_{t} = \max \{ \tilde{Y}_{t-1}^{\cA}  \, , \, 0 \} +\ell_t^\cA, \quad   \tilde{Y}_{0}^{\cA}:=0.
\end{align}
Moustakides  \cite{moust1}  showed that, for any given $\gamma\geq 1$,  $S_b^{\cA}$ optimizes Lorden's criterion, $\cJ_{\cA}$, in the class $\ccg$, when threshold $b$ is chosen so that $\Exp_{\infty}[S^{\cA}_{b}] = \gamma$. Earlier,    Lorden \cite{lorden1} had
established  the \textit{first-order} asymptotic optimality of the CUSUM rule
 by showing that
\begin{align} \label{eq:kl_distance}
    \inf_{T \in \ccg} \cJ_{\cA} [T]  & \sim  \frac{\log \gamma} { I_{\cA}}, \quad  I_{\cA}:= \Exp_{0}^{\cA} [ \ell_1^\cA],
\end{align}
under the assumption that $\ell_1^{\cA}$ is $\Pro_0^{\cA}, \Pro_\infty$- integrable. However, the exact optimality of the CUSUM rule can be used to obtain a  second-order characterization of the optimal performance. Indeed,   if additionally we have $ \Exp_{0}^{\cA} [ (\ell_1^\cA)^{2}]<\infty$, it is well known  (see, e.g., \cite{TNB_book2014}, \cite{wood}) that
\begin{align*}
\Exp_0^{\cA}[S_b^{\cA}]= \frac{b}{I_{\cA}} + \Theta(1),
\end{align*}
where $\Theta(1)$ is a term that  is bounded from above and below as $b \rightarrow \infty$. Moreover,  from \cite{khan} it follows that
$ \Exp_\infty[S_b^{\cA}] \sim e^b / v_{\cA}$ as $b \rightarrow \infty$, where $v_\cA$ is a model-dependent, renewal-theoretic constant. Therefore, we conclude that
$$
 \inf_{T \in \ccg} \cJ_{\cA} [T]  =  \frac{\log \gamma} { I_{\cA}} + \Theta(1).
$$
We summarize these characterizations of first and second order asymptotic optimality in the following Lemma.

\begin{lemma}
     \begin{enumerate}
        \item[(i)]  Suppose that  $\ell_1^{\cA}$ is  integrable with respect to $\Pro_\infty$ and $\Pro_0^\cA$ for every $\cA \in \cP$ and  that,  for any given $\gamma \geq 1$, a detection rule $T^{*}$ can be designed  so that
    $T^{*} \in \ccg$.   If
            \begin{align*}
                \cJ_{\cA} [T^{*}]  \leq   \frac{\log \gamma} {I_{\cA}} (1+o(1)), \quad \; \forall \;  \cA \in \cP,
            \end{align*}
            then $T^{*}$ is  first-order asymptotically optimal  with respect to class $\cP$.
        \item[(ii)] Suppose  additionally that $\ell_1^{\cA}$ is  square-integrable with respect to $\Pro_0^\cA$ for every $\cA \in \cP$. If
            \begin{align*}
                \cJ_{\cA} [T^{*}]  \leq   \frac{\log \gamma} {I_{\cA}}  + \calo (1),
                \quad \; \forall \;  \cA \in \cP,
           \end{align*}
                 then $T^{*}$ is  second-order asymptotically optimal  with respect to class $\cP$.
    \end{enumerate}
\end{lemma}

\noindent \textit{Remark:} It is useful to stress that the CUSUM statistic is often defined in the literature as  
\begin{align} \label{eq:cusum:stat_pos}
    Y^{\cA}_{t} := Z_t^{\cA} - \min_{0 \leq s \leq t} Z_{s}^{\cA},
\end{align}
or equivalently through the following  recursion
\[
    Y^{\cA}_{t} = \max\left\{ Y_{t-1}^{\cA} + \ell_t^\cA \, , \,  0 \right\},\quad
    Y_{0}^{\cA}:=0.
\]
Thus,  $Y^{\cA}_{t}$ differs from $\tilde{Y}^{\cA}_{t}$  only when the latter is negative, and the two statistics lead to the same stopping time when the threshold $b$ is positive. However,  with the non-negative version of the CUSUM statistic it is not in general possible to satisfy the false alarm constraint for any $\gamma>1$. Since in this work we focus on asymptotic optimality results,  we will use both versions of the CUSUM statistic, depending on which one is more technically convenient each time.

\subsection{Generalized and Mixture based CUSUM rules} \nolabel{sec:glrcusum}

We now construct various multichannel detection rules that will be shown to be second-order asymptotically optimal with respect to any given class $\cP$. In order to do so,  suppose for the moment  that the affected subset is   $\cA \in \cP$. At some time $t >\nu$, the likelihood ratio of $\Pro_\nu^{\cA}$ versus $\Pro_\infty$ is given by
\begin{align*} \nolabel{eq:lr}
    \frac{\Lambda_{t}^{\cA} }{  \Lambda_{\nu}^{\cA} } =
    \frac{d\Pro_\nu^{\cA}}{d\Pro_\infty} (\cFt) =
     \prod_{s=\nu+1}^{t} \exp (\ell_s^{\cA} ) ,
\end{align*}
and the corresponding log-likelihood ratio  takes the form
\[
    Z_{t}^{\cA} - Z_{\nu}^{\cA}
    = \log \frac{\Lambda_{t}^{\cA} }{  \Lambda_{\nu}^{\cA} }
    =\sum_{s=\nu+1}^{t} \ell_s^{\cA},
\]
where $\Lambda_t^{\cA}$, $Z_t^{\cA}$ and $\ell_t^{\cA}$ are defined in \eqref{eq:likelihood_ratio} and  \eqref{eq:log-likelihood_ratio}, respectively.  Maximizing with respect to both the unknown change-point, $\nu$, and the unknown subset, $\cA$,  suggests raising an alarm when the statistic
\[
    \max_{0 \leq s < t}  \max_{\cA \in \cP} (Z_{t}^{\cA} - Z_{s}^{\cA})
\]
becomes larger than some  threshold $b$.  Interchanging the two maximizations reveals that this detection statistic can be equivalently expressed as
\begin{align*} \nolabel{eq:represent}
   \max_{\cA \in \cP}   \max_{0 \leq s < t} (Z_{t}^{\cA} -Z_{s}^{\cA})
 =  \max_{\cA \in \cP} \tilde{Y}_{t}^{\cA} ,
\end{align*}
where $\tilde{Y}^{\cA}$ is defined in \eqref{eq:cusum:stat_neg}. Therefore, this rule stops   the first time $t$  there is a subset $\cA$ whose corresponding
CUSUM statistic  $\tilde{Y}_t^{\cA}$ exceeds threshold $b$.  If  we  only consider positive thresholds,  we may equivalently consider the  detection statistic $\max_{\cA \in \cP}  Y_{t}^{\cA}$, where   $Y^{\cA}$ is defined in \eqref{eq:cusum:stat_pos} (see the remark in the end of the previous subsection).

We can  generalize this scheme by allowing the thresholds that correspond to the various CUSUM statistics to differ. Thus, if
$\{b_\cA, \cA \in \cP\}$ is a family of  positive, constant thresholds, we obtain
  the  following  detection rule, to which we refer  as GLR-CUSUM:
\begin{align} \label{eq:glr_cusum}
     \min_{\cA \in \cP} S^{\cA}_{b_\cA }
        &= \inf \left \{ t \in \bN:  \max_{\cA \in \cP}  \, (Y_{t}^{\cA} - b_\cA)  \geq 0 \right\},
\end{align}
(Recall that  $S^{\cA}_{b}$ is the CUSUM rule defined in \eqref{eq:cusum}). We will show in Proposition  \ref{prop:epic_fail} that some intuitively reasonable threshold specifications, such as selecting each threshold $b_{\cA}$ proportionally to the Kullback-Leibler divergence $I_\cA$, may fail to guarantee even the first-order asymptotic optimality of  GLR-CUSUM.   For this reason, in what follows we restrict our attention to thresholds of the form
\begin{align*} \nolabel{select}
    b_\cA=b - \log p_{\cA} ,
\end{align*}
where  $b$ is a positive threshold that is determined by the false alarm constraint and  depends on $\gamma$, whereas  $\{p_\cA, \cA \in \cP\}$ are constants (\textit{weights}) that do \textit{not} depend on $\gamma$ and satisfy
\[
    p_{\cA} >0 \quad \! \forall \; \cA \in \cP \quad \text{and} \quad \sum_{\cA \in \cP} p_{\cA}=1.
\]
With this threshold specification the  GLR-CUSUM in \eqref{eq:glr_cusum}  takes the form
\begin{align} \label{eq:glr_cusum:cusum}
        S_b &:= \inf \left \{ t \in \bN:  \max_{\cA \in \cP} \left(Y_{t}^{\cA}  +\log p_\cA \right) \geq b \right\}
\end{align}
and may  be equivalently expressed as follows
\begin{align}
    S_b &= \inf \left \{  t \in \bN:  \max_{0 \leq s \leq t}  \max_{\cA \in \cP} \left(Z_{t}^{\cA} - Z_{s}^{\cA} +\log p_\cA \right) \geq b \right\} \nonumber \\
        &= \inf \left \{ t \in \bN:   \max_{0 \leq s \leq t}   \max_{\cA \in \cP} p_{\cA} \exp (Z_{t}^{\cA} - Z_{s}^{\cA})  \geq  e^{b} \right\}.    \label{eq:glr_cusum:llr}
\end{align}
When the $p_{\cA}$'s are identical,  we will refer to $S_b$ as the \emph{unweighted} GLR-CUSUM. However, the introduction of weights allows us to present the GLR-CUSUM together with competitive \emph{mixture-based} CUSUM tests. Specifically,  if we do not maximize but average  with respect to the affected subset in \eqref{eq:glr_cusum:llr}, we obtain  the following stopping time:
\begin{align} \label{eq:mixture_cusum:llr}
    \bar{S}_b  &:= \inf \left \{  t \in \bN:   \max_{0 \leq s \leq t}   \sum_{\cA \in \cP} p_{\cA} \exp (Z_{t}^{\cA} - Z_{s}^{\cA})  \geq  e^{b} \right\} .
\end{align}
 If we further  interchange maximization and summation in the detection statistic in  \eqref{eq:mixture_cusum:llr}, we obtain  the alternative detection statistic
$
    \sum_{\cA \in \cP} p_{\cA}  \exp( Y_t^{\cA} ).
$
For technical convenience, we replace  $Y^{\cA}$  with $\tilde{Y}^{\cA}$ and set
\begin{align} \label{eq:mixture_cusum:cusum}
    \tilde{S}_b := \inf \left \{  t \in \bN:    \sum_{\cA \in \cP} p_{\cA} \,
        \exp( \tilde{Y}_t^{\cA} )   \geq  e^{b} \right\} .
\end{align}

To sum up, we have introduced three distinct multichannel detection rules: the  GLR-CUSUM, $S_b$, which can be equivalently defined by either \eqref{eq:glr_cusum:cusum} or  \eqref{eq:glr_cusum:llr},  and two \textit{mixture-based} CUSUM rules, $\bar{S}_b$ and $\tilde{S}_b$, defined in \eqref{eq:mixture_cusum:llr} and \eqref{eq:mixture_cusum:cusum} respectively.  In the following theorem,  whose proof can be found in the \hyperref[appen]{Appendix}, we establish their second-order asymptotic optimality property with respect to an arbitrary class $\cP$.

\begin{theorem} \label{thm:second-order:glr_mix_cusum}
 For an arbitrary class $\cP$,
    \begin{enumerate}
        \item[(i)] if   $b =\log \gamma$,  then    $S_b, \bar{S}_b, \tilde{S}_b \in  \ccg$;
        \item[(ii)] if $\ell_1^{\cA}$ is  square-integrable with respect to $\Pro_0^\cA$ for every $\cA \in \cP$,  then $S, \bar{S}, \tilde{S}$ are second-order asymptotically optimal with respect to $\cP$.
    \end{enumerate}
\end{theorem}

Although it was not needed in the  proof of the previous  theorem, it is important (and useful for computational purposes) to note that the worst-case scenario for all the above detection rules is  when  $\nu=0$. This is the content of the following proposition, whose proof can be found in the \hyperref[appen]{Appendix}.

\begin{proposition} \label{prop1}
    For any $b > 0$ and $\cA \in \cP$ we have
    \[
      i) \;  \cJ_{\cA}[S_b] = \Exp_{0}^{\cA}[S_b] , \quad ii) \;
        \cJ_{\cA}[\bar{S}_b] = \Exp_{0}^{\cA}[\bar{S}_b]
      ,  \quad   iii)    \;  \cJ_{\cA}[\tilde{S}_b] = \Exp_{0}^{\cA}[\tilde{S}_b] .
    \]

   \end{proposition}

\subsection{Implementation} \label{implement}

The two representations of the  GLR-CUSUM, \eqref{eq:glr_cusum:cusum} and
\eqref{eq:glr_cusum:llr}, suggest two possible approaches regarding its implementation. According to the first one,  at each time $t$  we compute and maximize over the CUSUM statistics,  $Y_{t}^{\cA}$, $\cA \in \cP$. This approach requires
running  $|\cP|$ recursions and is  convenient when the cardinality of $\cP$ is small. This is for example the case when only one component of the parameter vector $\theta$ is affected by the change  ($\cP=\cP_{1}$), in which case
$|\cP|=N$. However, such an implementation may not be feasible for large  values of $N$ when we have  complete ignorance regarding the affected subset ($\cP=\overline{\cP}_{N}$).  (This approach  also applies to the  mixture-based CUSUM rule  defined in  \eqref{eq:mixture_cusum:cusum}, which sums over the exponents of all CUSUM statistics).

According to the second approach,  at each time $t$ we need to compute and maximizer over  the statistics
$ Z_{t}^{\cA} - Z_{s}^{\cA} +\log p_\cA$, $0 \leq s <t$, $\cA \in \cP$.
In the next section  we will see  that this computation is simplified
in the special case that we have  independent sensors and only a subset of which is affected by the change. Moreover, in the next lemma we show that, without any loss,   we can always restrict the maximization to an adaptive time-window. Since a similar approach applies to the mixture-based CUSUM  defined in  \eqref{eq:mixture_cusum:llr}, we state the following lemma in some generality.


\begin{lemma} \label{lem2}

Let $\cP$ be an arbitrary class and  consider the  following sequence of stopping times
\begin{align} \label{eq:window:mei}
\begin{split}
    r_{n + 1}
    &= \stopset{t > r _n }{Z_t^{\cA } = \min _{0 \leq s \leq t} Z_s^{\cA } ,\quad \text{for all $\cA \in \cP $}};   \; r_0 = 0.
  \end{split}
\end{align}
Then for any function  $G : \Reals^{|\cP |} \to \Reals $ that is non-decreasing in each of its arguments,
\begin{align*}
    \stopset{t \geq 0}{\max _{0 \leq s \leq t} G\left((Z_t^{\cA } - Z_s^{\cA })_{\cA \in \cP }\right) \geq b}& \\
    = \stopset{t \geq 0}{\max _{r(t) \leq s \leq t} G\left((Z_t^{\cA  } - Z_s^{\cA })_{\cA \in \cP }\right) \geq b}&,
\end{align*}
where  $(\cdot )_{\cA \in \cP }$ denotes a $\cP $-tuple and     $r(t) = \max \cset[:]{r_n}{r_n \leq t}$.
\end{lemma}

\begin{proof}
 Fix $t > 0$.  Since $G$ is non-decreasing in each of its  arguments, it suffices to show that for every  $0 \leq s\leq r(t)$ and  $\cA \in \cP$  we have $ Z_t^{\cA } - Z_{s}^{\cA } \leq Z_t^{\cA } - Z_{r(t)}^{\cA }$, or equivalently
  $Z_{r(t)}^{\cA } \leq  Z_{s}^{\cA }$, which follows from the
 definition of $r(t)$.
\end{proof}

\section{A special multisensor sequential change-detection problem} \label{sec:special_problem}

\subsection{A special case of the general framework}

In this section,  we focus on the special case that the change affects the distribution of only  a subset of  sensors. In order to be more specific, let us  recall that at each time $t$ we observe a random vector $X_t=(X_t^1, \ldots, X_t^{K})$ with joint density
$X_t \sim f( \cdot|\theta_t)$, where  $\theta_t=(\theta_t^1, \ldots, \theta_t^N)$ is an unobserved parameter vector. Let  $f_k$  be  the marginal density of the $k^{\text{th}}$ stream, i.e., $X_t^k \sim f_k( \cdot| \theta_t)$.  In what follows, we assume that the dimensions of  $X_t$ and   $\theta_t$ coincide, i.e., $N=K$, and that each marginal density $f_k( \cdot| \theta_t)$  is determined only by the $k^{\text{th}}$ component of  $\theta_t$, $\theta_t^k$. Specifically, we assume that
\begin{align*}
    f_k(\cdot|\theta_t^k=0) &=h_k, \quad
    f_k(\cdot|\theta_t^k=1) =g_k ,
\end{align*}
where $h_k$ and $g_k$ are completely specified densities with respect to some $\sigma$-finite measure  $\lambda_k$. Note that  $h_k$ and $g_k$ do not need to belong in the same parametric family,  our only assumption is that their Kullback-Leibler information number,
\begin{align} \label{local_KL}
    I_k:= \int \log \left( \frac{ g_{k} } { h_{k} }  \right) \; g_{k} \;  d\lambda_{k},
\end{align}
is positive and finite.  Therefore, in this setup  we have  a change in the marginal distributions of a subset of sensors, while the remaining ones remain completely unaffected, and the change detection problem \eqref{eq:setup:param_change} takes the form
\begin{align} \label{eq:setup:marginal_change}
    \begin{split}
        X_t^k &\sim h_k , \quad  \forall \;  1 \leq k \leq K,  \quad t \le \nu,  \\
        X_t^k &\sim \begin{cases}
            g_k , & \quad k \in \cA, \\
            h_k , &  \quad k \notin \cA,
        \end{cases}
        \quad t > \nu,
    \end{split}
\end{align}
where $\cA \subset \{1, \ldots, K\}$ is a subset of sensors that belongs to some class $\cP$. When in particular $\cP=\cP_L$ (resp. $\cP=\overline{\cP}_L$), the change affects exactly (resp. at most)  $L$  sensors, where the classes  $\cP_L$ and $\overline{\cP}_L$ are defined in \eqref{eq:class_arl}.

In the literature of the multisensor quickest detection problem it is typically assumed that observations from different streams are independent. However, this is not necessary for the results of the previous section to hold. This is  illustrated by the following example. \\

\noindent \textit{Example: Correlated normal streams}

Let $\Sigma$ be an invertible covariance matrix of dimension $K$ with
diagonal $(\sigma_1^{2}, \ldots, \sigma_K^2)$ and let  $\mu_1, \ldots, \mu_K$ be non-zero constants. For every non-empty subset $\cA \subset \{1, \ldots, K\}$ we define  the $K$-dimensional vector
$\mu_\cA:= ( \mu_1^\cA, \ldots, \mu_K^\cA)$ such  that
\[
    \mu_k^\cA :=
    \begin{dcases}
       \mu_k  , &  \;  k \in \cA, \\
        0 ,     &  \;  k \notin \cA,
    \end{dcases}
\]
and we assume that under $\Pro_\nu^\cA$ we have
\[
    X_t    \sim
    \begin{dcases}
      \cN(0, \Sigma),        &t \le \nu, \\
      \cN(\mu_\cA, \Sigma),  &t > \nu .
    \end{dcases}
\]
Then, this is a special case of  \eqref{eq:setup:marginal_change}  with  $h_k= \cN(0, \sigma_k^2)$ and $g_k= \cN(\mu_k, \sigma_k^2)$, for which
\[
 \ell_t^\cA= \theta_{\cA} \cdot X_t- \frac{1}{2} \theta_{\cA} \cdot \mu_{\cA}, \quad \quad  \theta_{\cA} := \Sigma^{-1} \mu_{\cA},
\]
where   $\ell_t^\cA$ is   the log-likelihood ratio defined in \eqref{eq:log-likelihood_ratio}.


\subsection{The case of independent sensors}
We now restrict ourselves to the case of independent sensors. Specifically, we assume that the local filtrations $\{\cF_t^1\}, \ldots, \{\cF_t^K\}$ are independent, where $  \cFt^k:=\sigma(X_s^k, \; 1 \leq s \leq t)$, $1 \leq k \leq K$.   Under this assumption, the log-likelihood ratio   statistic $\ell_t^{\cA}$, defined in \eqref{eq:log-likelihood_ratio},   admits the following decomposition
\begin{align} \label{eq:indep_llr_decomposition}
    \ell_t^{\cA} =  \sum_{k \in \cA} \ell_t^k, \quad
    \quad  \ell_t^k:= \log \frac{g_k}{h_k} (X_t^k),
\end{align}
which  implies that the likelihood and log-likelihood ratio statistics,  $\Lambda_t^{\cA}$ and  $Z_t^{\cA}$,  defined in \eqref{eq:likelihood_ratio}, take the form
\begin{align*}
    \Lambda_t^{\cA} &=  \prod_{k \in \cA} \Lambda_t^k, \quad
    \Lambda_t^k := \Lambda_{t-1}^{k} \, \exp\{ \ell_t^k \} , \quad \Lambda_0^k=1 , \\
          Z_t^{\cA} &=  \sum_{k \in \cA} Z_t^k, \quad Z_t^k:= Z_{t-1}^{k}+  \ell_t^k, \quad  Z_0^k=0.
\end{align*}
Moreover, we have $I_{\cA}=\sum_{k \in \cA} I_k$, where  $I_{\cA}$ is the Kullback-Leibler information number  defined in \eqref{eq:kl_distance} and $I_k$ the local Kullback-Leibler information number defined in \eqref{local_KL}.  If we further select weights of the form
\begin{align} \label{eq:weight}
    p_{\cA} &=\frac{ p^{|\cA|} }{\sum_{\cB \in \cP} p^{|\cB|}  },
\end{align}
where $p$ is some arbitrary positive parameter,  then from  \eqref{eq:glr_cusum:llr}  and \eqref{eq:mixture_cusum:llr} it follows that
the GLR-CUSUM  and the  mixture-based CUSUM rule, $\bar{S}_b$, can be expressed as follows:
\begin{align}
    S_b &:= \inf \left \{ t \in \bN:  \max_{0 \leq s \leq t} \, \max_{\cA \in \cP}  \sum_{k \in \cA} \left(Z_{t}^{k}-Z_{s}^{k} +\log p \right)  \geq
    b +\log (  C(\cP) )  \right\},  \label{eq:glr_cusum:factorized} \\
       \bar{S}_b  &:= \inf \left \{t \in \bN:   \max_{0 \leq s \leq t}   \sum_{\cA \in \cP} p^{|\cA|} \prod_{k \in \cA} \exp (Z_{t}^{k} - Z_{s}^{k})  \geq  e^{b} \, C(\cP)  \right\}, \label{eq:mixture_cusum:factorized}
\end{align}
where $C(\cP):=\sum_{\cB \in \cP} p^{|\cB|}$.  In the following proposition, whose proof is presented in the \hyperref[appen]{Appendix},  we show that the latter expressions can be  further simplified for classes of the form  $\cP_{L}$ and  $\overline{\cP}_{L}$.  In order to do so,   for any  $0 \leq s \leq t$ we set
\[
    Z_{s:t}^{k} := Z_{t}^{k}-Z_{s}^{k}, \quad   1 \leq k \leq K,
\]
and  we introduce the corresponding order statistics
$
    Z_{s:t}^{(1)} \geq \ldots \geq Z_{s:t}^{(K)}.
$




\begin{proposition} \label{prop2}
Consider the change detection problem \eqref{eq:setup:marginal_change} and suppose that the independence assumption  \eqref{eq:indep_llr_decomposition}   holds.

(i)  The GLR-CUSUM in  \eqref{eq:glr_cusum:factorized} takes the  form
    \begin{align}
& \inf \Bigl\{ t \in \bN : \max _{0 \leq s \leq t}  \sum _{k = 1}^L  Z_{s:t}^{(k)}  \geq b +\log |\cP_{L}|    \Bigr\}  \quad  \text{when} \quad   \cP=\cP_{L},
 \label{eq:glr_cusum:impl_exact} \\
& \inf \left\{ t \in \bN : \max _{0 \leq s \leq t}  \sum _{k = 1}^L
            \left( Z_{s:t}^{(k)} + \log p \right)^{+} \geq
            b+ \log \left( \sum_{k=1}^{L} |\cP_k| \,  p^{k} \right) \right\}  \quad  \text{when} \quad   \cP=\overline{\cP}_{L}. \label{eq:glr_cusum:impl_at_most}
    \end{align}

(ii)     When $\cP=\overline{\cP}_{K}$, the mixture rule $\bar{S}_b$  in
    \eqref{eq:mixture_cusum:factorized} takes the form
    \begin{align} \label{eq:mixture_cusum:bayes}
        \hat{S}_{\hat{b}} (\pi)  &:= \inf \left \{ t \in \bN:
            \max_{0 \leq s \leq t}  \prod_{k=1}^{K}   \left[ 1-\pi +\pi \, \exp(Z_{t}^{k}-Z_{s}^{k})  \right]  \geq \exp\{\hat{b}\}  \right\},
    \end{align}
    where  $\pi=p/(1+p)$ and
    \[
        \exp\{\hat{{b}}\}:=    (1-\pi)^{-K}  \left( \exp\{b\} \sum_{\cB \in \cP} p^{|\cB|}+1 \right).
    \]

\end{proposition}

\noindent \textit{Remark:}  From Lemma \ref{lem2} we know that we can replace $\max_{0\leq s \leq t}$  in  the above detection rules 
 by  $\max_{r(t) \leq s \leq t}$, where $r(t)=\max\{n: r_{n} \leq t\}$ and the sequence of stopping times $(r_n)$ are defined in  \eqref{eq:window:mei}. In the case of independent sensors that consider in this section we have  $Y_t^\cA\leq \sum_{k \in \cA} Y_t^k$ for every $t$ and $\cA$,  which implies that for a class of the form $\overline{\cP}_{L}$ the times $(r_n)$ can be defined as follows:
\begin{align*} \nolabel{eq:window:mei2}
    r_{n}
   &= \stopset{t > r _{n-1} }{Z_t^{k } = \min _{0 \leq s \leq t} Z_s^{k}, \quad \forall \,  1 \leq k \leq K },
\end{align*}
where $r_0 := 0$. From Proposition 1 in \cite{mei_bio} it follows that $(r_n-r_{n-1})_{n \in \bN}$
are  iid random variables with  finite expectation under $\Pro_\infty$, which however  grows exponentially in $K$. This implies that the above window may not be very useful  for  computational purposes.  A more appropriate approach for practical implementation when $K$ is large is  to use instead the window
$
    \sigma(t):= \max\{ \sigma_n: \sigma_n \leq t\},
$
where 
\begin{align*} \nolabel{eq:rlr_sigma}
    \sigma_n:=\inf \left\{t> \sigma_{n-1}: Z_t^k < Z_{\sigma_{n-1}}^k, \quad \forall \; 1 \leq k \leq K \right\}; \quad \sigma_0:=0.
\end{align*}
Again,  $(\sigma_n-\sigma_{n-1})_{n \in \bN}$ is a sequence of iid random variables with  finite expectation under $\Pro_\infty$, which however seems to  grow logarithmically in $K$ (based on empirical observations). However, we should emphasize that  the resulting detection rules are not equivalent to the original ones and their asymptotic performance requires separate analysis.  For similar adaptive window-limited modifications of CUSUM-type rules  we refer to Yashchin \cite{yash1, yash2, yash3}. \\


\noindent \textit{Remark:} When the weights are selected according to \eqref{eq:weight},  all subsets of the same size have the same weight. It is straightforward to generalize the previous results when  $p_{\cA}$ is proportional  to  $\prod_{k \in \cA}  p_k$, where  $p_k$, $1 \leq k \leq K$ are arbitrary positive parameters. However, setting $p_k=p$ offers an intuitive  way of selecting  the parameter $\pi$,  or equivalently $p$, in the mixture-based CUSUM,   $\hat{S}_{\hat{b}} (\pi) $, in \eqref{eq:mixture_cusum:bayes} Indeed,  $\hat{S}_{b}(\pi)$ can be obtained by  repeated application of the one-sided sequential  test
\begin{align} \label{eq:mixture_sprt:bayes}
    \begin{split}
        \hat{T}_{b}(\pi) &:= \inf \left\{ t \in \bN  :  \hat{\Lambda}_{t}  \geq e^{b} \right\},\quad \text{where} \\
        \hat{\Lambda}_{t} &:= \prod_{k=1}^{K}
            \left(1 - \pi + \pi \,  \Lambda_{t}^{k} \right).
    \end{split}
\end{align}
This is the one-sided SPRT for testing $\Pro_{\infty}$ against the auxiliary  probability measure
\[
    \hat{\Pro} =\times_{k=1}^{K} \left[(1-\pi) \Pro_\infty^k+ \pi \Pro_0^k \right],
\]
according to which the density in sensor $k$ is $g_k$ with probability $\pi$ and $h_k$ with probability $1-\pi$.   This implies that  the parameter $\pi$ in
$\hat{S}_{\hat{b}} (\pi) $  can be interpreted as the proportion of affected sensors and suggests setting  $\pi=L/K$ when we know in advance that exactly $L$ sensors are affected $(\cP=\cP_{L})$.   On the other hand, if  it is known  that \textit{at most} $L$ sensors are  affected   $(\cP=\overline{\cP}_{L})$, a reasonable default choice seems to be $\pi=L/(2K)$. \\



\noindent \textit{Remark:}   It is  interesting to  compare the detection rule $\hat{S}_b$,  defined in  \eqref{eq:mixture_cusum:bayes}, with the detection rule
\begin{align} \label{eq:mixture_cusum:siegmund}
    \check{S}_{b} (\pi) &:= \inf \Bigl\{ t \in \bN : \max _{0 \leq s \leq t}
    \sum_{k=1}^{K}     \log \left(1 - \pi + \pi \, \exp\{ (Z_{t}^{k}-Z_s^k)^{+} \}  \right)\geq b\Bigr\},
\end{align}
which  was proposed by Xie and Siegmund \cite{xie}. In the following  theorem, whose proof is presented in the \hyperref[appen]{Appendix},  we show that the latter is also second-order asymptotically  optimal  with respect to $\overline{\cP}_{K}$ for any choice of $\pi$ in $(0,1]$, a result that explains the (empirically observed in \cite{xie})  robustness of this rule with respect to $\pi$.

\begin{theorem} \label{thm:second-order:bayes_sieg_cusum}
Consider the change detection problem \eqref{eq:setup:marginal_change} and suppose that the independence assumption  \eqref{eq:indep_llr_decomposition}   holds.   Consider the detection rules  $\hat{S}_b(\pi)$ with some $\pi \in (0,1)$ and $\check{S}_b(\pi)$ with some $\pi \in (0,1]$, defined in  \eqref{eq:mixture_cusum:bayes} and  \eqref{eq:mixture_cusum:siegmund}, respectively.  Then,
    \begin{enumerate}
        \item[(i)]   $\hat{S}_b(\pi) \in \ccg$ if  $b=\log \gamma$ and
            $\check{S}_b(\pi) \in \ccg$ if  $b=\log \gamma+ \log |\overline{\cP}_{K}|$;
        \item[(ii)]  if  also $\ell_1^\cA$ is square-integrable under $\Pro_0^{\cA}$ for every $\cA \in \overline{\cP}_{K}$,   both $\check{S}_b(\pi)$ and $\hat{S}_b(\pi)$  are  second-order asymptotically optimal with respect to  $\overline{\cP}_{K}$.
    \end{enumerate}
\end{theorem}

\noindent \textit{Remark:}
We  stress that the parameter $\pi$ in the detection rule proposed by   Xie and Siegmund \cite{xie}, $\check{S}_{b} (\pi) $, \textit{cannot} be interpreted as the proportion of affected sensors. Indeed, if we set  $\pi=1$ in \eqref{eq:mixture_cusum:siegmund} we recover the  unweighted GLR-CUSUM  in the case of complete uncertainty, that is the detection rule  obtained by setting  $L=K$ and $p=1$ in \eqref{eq:glr_cusum:impl_at_most}.  On the other hand, if we set $\pi=1$ in $\hat{S}_b(\pi)$, defined in \eqref{eq:mixture_cusum:bayes},  we recover the optimal  CUSUM rule in the case that all sensors are affected, which is exactly what one would expect if $\pi$ is to be interpreted as the proportion of affected sensors.

\subsection{Scalable schemes}

We close this section by highlighting the connection between the GLR-CUSUM and the  SUM-CUSUM proposed by Mei  in \cite{mei_bio}. Recall that the detection statistic of the  \textit{unweighted GLR-CUSUM  in the case of complete ignorance}, that is the detection rule  obtained by setting  $L=K$ and $p=1$ in \eqref{eq:glr_cusum:impl_at_most}, is
\[
    \max _{0 \leq s \leq t} \sum _{k = 1}^K  \left( Z_{t}^{k}-Z_{s}^{k} \right)^{+}.
\]
 If we interchange max and sum  in  this statistic, we obtain  the sum of the local CUSUM statistics, $\sum _{k = 1} ^{K}  Y_t ^{k}$.  This is the detection statistic of  SUM-CUSUM, which was shown in \cite{mei_bio} to be \textit{first-order} asymptotically optimal with respect to $\overline{\cP}_{K}$, i.e., in the case of complete ignorance. The main advantage of this rule is that it is much  easier to implement than the corresponding GLR-CUSUM and  mixture-based CUSUM rules, as it requires running only $K$ recursions.  However, it is reasonable to expect that SUM-CUSUM should be less efficient, since  it is only first-order asymptotically optimal. This intuition will be corroborated by the results of a simulation study in Section \ref{sec:simulation}. \\

\noindent \textit{Remark:} If we  interchange, in a similar way,  max and product in the mixture-based CUSUM rule \eqref{eq:mixture_cusum:bayes}, we obtain
\begin{align*} \nolabel{eq:mixture_cusum:mix_mei}
    M_b(\pi) &:= \inf \Bigl\{ t \in \bN :
        \sum_{k=1}^{K}   \log \left(1 - \pi + \pi \, \exp\{Y_t^k\}  \right)\geq b \Bigr\}.
\end{align*}
This is another scalable rule  comparable with SUM-CUSUM, however  the parameter $\pi$ can no longer be interpreted as proportion of affected sensors.\\

\noindent \textit{Remark:}  From Mei \cite{mei_sympo} it follows that  the detection rule
\begin{align} \label{eq:mixture_cusum:mei}
    M_b(L) &= \inf \left\{ t\in \bN : \sum _{k = 1} ^{L} Y_t ^{(k)}  \geq b \right\},
    \end{align}
where
$
Y_t ^{(1)} \ge Y _t ^{(2)} \ge \cdots \ge Y_t ^{(K)},
$
is  \emph{first-order} asymptotically optimal with respect to $\overline{\cP}_{L}$ (at most $L$ sensors affected).  This rule should be compared with the  corresponding GLR-CUSUM and mixture-based CUSUM
that correspond to class $\overline{\cP}_{L}$.


\section{Multichart CUSUM } \label{sec:one_affected}
In this section we remain  in the change-detection problem  \eqref{eq:setup:marginal_change}, where the change affects
only a subset of sensors. Moreover,  we assume that the various sensors
are independent, i.e.,  \eqref{eq:indep_llr_decomposition} holds. Our  focus will be on the  GLR-CUSUM, defined in \eqref{eq:glr_cusum}, in the special case that the change can affect exactly one sensor $(\cP= \cP_{1})$, where  it takes the form
\begin{align} \label{eq:multichart}
    S=\min_{1\leq k \leq K} S^{k}_{b_{k}},
\end{align}
each $b_k$ being a positive threshold.  
The detection rule \eqref{eq:multichart}  is also known  as  \textit{multichart} CUSUM  in the literature of  statistical  quality control and has been well  studied (see, e.g.,  \cite{blazek, blazek2}), especially  in the case that all thresholds are equal.
However, unless the sensors are homogeneous in the sense that $h_k=h$ and $g_k=g$, it is intuitively clear that it should be preferable to have unequal thresholds. The results of Section \ref{sec:general_problem} imply  that there is a large family of non-identical thresholds for which the multichart CUSUM is second-order asymptotically optimal. Specifically, if we set
\begin{align} \label{eq:select2}
    b_k=b -\log p_k, \quad 1 \leq k \leq K,
\end{align}
where $b$ is determined by the false alarm constraint and the $p_k$'s are arbitrary positive constants that add up to 1 and do not depend on $\gamma$, then from Theorem \ref{thm:second-order:glr_mix_cusum} it follows
the resulting detection rule is second-order asymptotically optimal for any selection of the $p_k$'s.  Our first goal in this section is to show that there are reasonable, alternative threshold specifications with which the multichart CUSUM may  fail to  be  even first-order asymptotically  optimal. The second goal is to compare various choices for the $p_k$'s when the thresholds are selected according to \eqref{eq:select2}.

In order to establish these results, we will rely on some well known facts about the performance characteristics of the  multichart CUSUM. Thus, let $T^{k}_{b} = \inf \left\{ t \geq 0 : Z^{k}_{t} \geq b \right\}$ and let $\eta^{k}_b$ be the  corresponding overshoot, i.e.,
$
   \eta^{k}_b:= Z_{T^{k}_b} ^{k} - b.
$
Since $Z^{k}$ is a random walk with a finite second moment, the limiting expected overshoot is well defined and we have
\begin{align*} \nolabel{eq:rd}
    \rho_{k} &:= \lim_{b \rightarrow \infty} \EV_{0}^{k}[ \eta^{k}_b].
\end{align*}
Moreover,  we have the following representation  (see \cite[p.32]{wood} or \cite[p.44]{TNB_book2014}) for the expected infinum of $Z^{k}$ under $\Pro_0^k$:
\[
    \beta_{k} := \EV_{0}^{k} \left[\inf_{t \in \bN} Z_{t}^{k} \right]= \frac{\Exp_0^k[ (Z_1^k)^2]}{2 I_k}-\rho_k.
\]
Finally, we introduce the Laplace transform of the limiting distribution of the overshoot under $\Pro_0^k$:
\begin{align*} \nolabel{eq:rd2}
    \delta_{k} := \lim_{b \rightarrow \infty} \EV_{0}^{k}[ e^{-\eta^{k}_b}].
\end{align*}
Then, it is well known (see, e.g., \cite[p.471]{TNB_book2014})
that  as $b_k \rightarrow \infty$
\begin{align} \label{eq:delay_high}
    \cJ_{k}[S] = \Exp_{0}^{k}[S] \simeq \frac{b_k +\rho_{k}+ \beta_{k}}{\KLI_{k}},
\end{align}
where  $x \simeq y$ means $x-y=o(1)$. On the other hand,   under $\Pro_{\infty}$ we have  (see, e.g., \cite{khan},   \cite[p.467]{TNB_book2014})
\begin{align}  \label{eq:alarm2222}
    \EV_{\infty}[ S] &\sim \frac{1}{ \sum_{k=1}^{K} e^{-b_k}\KLI_{k} \delta_{k}^{2} }.
\end{align}
From  \eqref{eq:delay_high} it is clear that  if  the $b_k$'s are chosen proportionally to the Kullback-Leibler information numbers, i.e., $b_k \propto I_k$,  then the  first-order asymptotic performance of the  multichart CUSUM under the various scenarios is equalized, i.e.,
$\cJ_{k}[S] \sim \cJ_{m}[S]$ for every  $ 1 \leq k \neq m \leq K$.
However,  as we show in the  following proposition,  with  this threshold specification the  multichart CUSUM loses even its first-order  asymptotic optimality property with respect to $\cP_{1}$, unless the $I_k$'s are identical.

\begin{proposition} \label{prop:epic_fail}
    Suppose that $b_{k}= c_{\gamma} I_{k}$ for every $1 \leq k \leq K$,  where $c_\gamma$ is a constant that does not depend  on $k$ and is chosen so that the false alarm constraint be satisfied with equality, i.e., $\Exp_{\infty}[S]=\gamma$. Then, for every $1 \leq k \leq K$ we have:
    \[
    \lim_{\gamma \rightarrow \infty}   \frac{\cJ_{k}[S]}{\inf_{T \in \ccg} \; \cJ_{k}[T]} =     \frac{\KLI_{k}}{\min_{1 \leq j \leq K} \KLI_{j}}.
    \]
\end{proposition}

\begin{proof}
    Setting  $b_{k}= c_\gamma I_{k}$ in \eqref{eq:delay_high} and \eqref{eq:alarm2222} we obtain that
    $\cJ_{k}[S]= c_{\gamma} + \calo(1), $  where $\calo(1)$ is an asymptotically bounded term   as $\gamma \rightarrow \infty$, and
    \begin{align}  \label{eq:aux:alarm}
      \gamma=  \EV_{\infty}[ S] &\sim \frac{1} { \sum_{k=1}^{K} e^{-c_{\gamma} I_k} \KLI_{k} \delta_{k}^{2} }.
        \end{align}
     Let $I^{*}:= \min_{1 \leq j \leq K} \KLI_{j} $. Then \eqref{eq:aux:alarm} becomes
    \begin{align*}  \nolabel{eq:alarm22}
      \gamma \sim \frac{e^{c_{\gamma} I^{*}} } { \sum_{k: I_{k}=I_{*}}   \KLI_{k} \delta_{k}^{2}+ \sum_{k: I_{k}>I_{*}}  e^{-c_{\gamma} (I_{k}-I_{*}) }  \KLI_{k} \delta_{k}^{2}},
    \end{align*}
    which implies that
    $
      c_{\gamma}=(\log \gamma) / I_{*}+ \calo(1)
    $
    and  completes the proof. \\
\end{proof}

Let us now focus on  specification  \eqref{eq:select2} for the thresholds.
While  we have seen in Theorem \ref{thm:second-order:glr_mix_cusum} that
setting $b = \log \gamma $ guarantees  $S \in \ccg$,  it is clear that this choice can be very conservative, as it does not take into account  the particular pre/post-change  distributions. The asymptotic expression \eqref{eq:alarm2222}  can be used to  provide  a  much more accurate approximation for $\Exp_{\infty}[S]$, even though it may not guarantee that  $\Exp_{\infty}[S] \geq \gamma$. Specifically,  from  \eqref{eq:delay_high} and  \eqref{eq:alarm2222} it follows that if we select threshold $b$ in \eqref{eq:select2} as
\begin{align} \label{eq:approxa}
    b = \log \gamma +
        \log \Bigl( \sum_{k=1}^{K} p_{k} \KLI_{k} \delta_{k}^{2} \Bigr),
\end{align}
then   $ \Exp_{\infty}[S] \sim \gamma$ and
\begin{align} \label{eq:tilde_delay}
    \cJ_{k} [S]     & \simeq \frac{  \log \gamma + \log \Bigl( \sum_{j=1}^{K} p_{j} \KLI_{j} \delta_{j}^{2} \Bigr) -\log p_{k}+ \rho_{k}+ \beta_{k}}{ \KLI_{k}}.
\end{align}
We can then use this high-order approximation in order to obtain an intuitive selection of the  $p_{k}$'s. Indeed, such a selection can be obtained if we
equalize, up to a first order,  the asymptotic \textit{relative} performance loss under each scenario. To be more specific,   let us  set
\[
   \bar{\cJ}_{k} [S]:= \frac{\cJ_{k}[S] -  \cJ_{k}[S^{k}]}{\cJ_{k}[S^k]}, \quad 1 \leq k \leq K ,
\]
 where $S^{k}$ is the optimal  CUSUM test for detecting a change  in sensor $k$.  Assuming that the thresholds for the multichart CUSUM are selected according to \eqref{eq:select2} and \eqref{eq:approxa} and that $S^{k}$ is also designed so that $\Exp_{\infty}[S^{k}] \sim \gamma$, then from  \eqref{eq:tilde_delay}  we  obtain
\begin{align} \label{eq:tilde_delay3}
    \bar{\cJ}_{k}[S] &\simeq
    \frac{C_{k}(p)}   { \log \gamma  +\rho_{k}+ \beta_{k}+ \log (\KLI_{k} \delta_{k})^{2} } ,
\end{align}
and consequently
\[
    \bar{\cJ}_{k}[S]  \sim \frac{C_{k}(p)}{\log \gamma},
\]
where  $C_k(p)$ is defined as follows:
\begin{align*} \nolabel{eq:C}
    C_{k}(p):=
    \log \left(  \frac{  \sum_{j=1}^{K} p_{j}  \, \KLI_{j} \delta_{j} ^{2} }
    { p_{k} \, \KLI_{k}  \delta_{k}^{2} }   \right).
\end{align*}
This reveals that  $\bar{\cJ}_{k} [S] \sim  \bar{\cJ}_{m} [S]$  for every $1 \leq k \neq m \leq K$ when
\begin{align} \label{eq:select_pk}
    p_{k} \propto \left( I_{k}  \delta_{k}^{2} \right)^{-1}.
\end{align}
In fact, the latter specification was  proposed  in \cite[p.223]{tarta_multi} on the basis of the observation that it guarantees
$
\Exp_{\infty}[S^k_{b_k}] \sim \Exp_{\infty}[S^m_{b_m}],
$
for every $1 \leq  k \neq  m \leq K$.



An alternative  specification for the $p_k$'s has been proposed in  \cite[p.467]{TNB_book2014}, according to which thresholds are selected so that
$
I_{k} \cJ_{k}[S] \simeq  I_{m} \cJ_{m}[S].
$
for every $1 \leq  k \neq  m \leq K$. From  \eqref{eq:tilde_delay} it is clear that this is the case  when
\begin{align*} \nolabel{eq:kl_select}
    p_{k} \propto  e^{\rho_{k}+\beta_{k}}.
\end{align*}
Finally, if we want to have larger thresholds for sensors with larger Kullback-Leibler information numbers, which is the underlying logic  in  the threshold specification of Proposition \ref{prop:epic_fail}, we can set
\begin{align*} \nolabel{eq:inverse_kl_pk}
    p_{k} \propto 1 / \KLI_{k}.
\end{align*}

In the remainder of this section we  compare the effect that the above mentioned  specifications have on the  asymptotic  relative performance loss of a multichart CUSUM whose thresholds are selected according to \eqref{eq:select2} and \eqref{eq:approxa}. Specifically, we assume that $h_{k}=\cN(0,1)$ and $g_{k}=\cN(\theta_k,1)$, where $\theta_k \neq 0$.
We  let  $K=2$, $\gamma=10^{3}$ and  we compute  $\bar{\cJ}_{1} [S]$ and  $\bar{\cJ}_{2} [S]$, based on the asymptotic approximation \eqref{eq:tilde_delay3}, fixing  the  post-change mean in the first sensor at  $\theta_{1}=1$, and letting  the post-change mean in the second sensor, $\theta_{2}$,  vary.   In this way, we  examine the relative performance loss  as the  relative magnitude of the anticipated change in the two sensors varies.  In the Gaussian case we have the following, easily computable expressions   (see, e.g., \cite[p.32]{wood})  for the renewal-theoretic quantities  that are included in this asymptotic approximation:
\begin{align*}
    \beta _{k} &= -  \sum_{n=1}^{\infty} \Bigl[ \frac{\theta_{k} }{\sqrt{n}} \, \phi ( c_{n,k} )
                      - \frac{\theta_{k}^{2}}{2} \, \Phi( -c_{n,k} )  \Bigr], \\
    \delta_{k} &= \frac{2}{\theta_{k}^{2}} \exp\Bigl\{\sum_{n=1}^{\infty} \frac{-2}{n} \, \Phi( -c_{n,k}) \Bigr\} ,
\end{align*}
where $c_{n,k}:= \theta_{k}\sqrt{n} / 2$.  The results of this computation are summarized  in Fig.~\ref{fig:theta_1_0_reloss}, where we see that selecting  the  $p_{k}$'s according to \eqref{eq:select_pk} leads to a  more robust behavior in comparison to the other specifications when $\theta_2$ differs significantly from $\theta_1$. Of course, all specifications
approach the behavior of  identical thresholds when  $\theta_2$ is close to $\theta_1$.

%
\begin{figure}[!htb]
    \centering
    \begin{subfigure}[b]{.5\linewidth}
        \centering
        \includegraphics[width=1.0\textwidth]{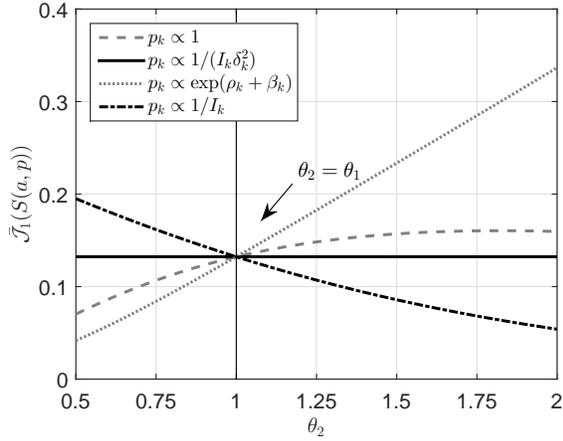}
        \caption{The change occurs in the \textit{first} sensor.}
        \nolabel{fig:theta_1_0_reloss:one}
    \end{subfigure}%
    \begin{subfigure}[b]{.5\linewidth}
        \centering
        \includegraphics[width=1.0\textwidth]{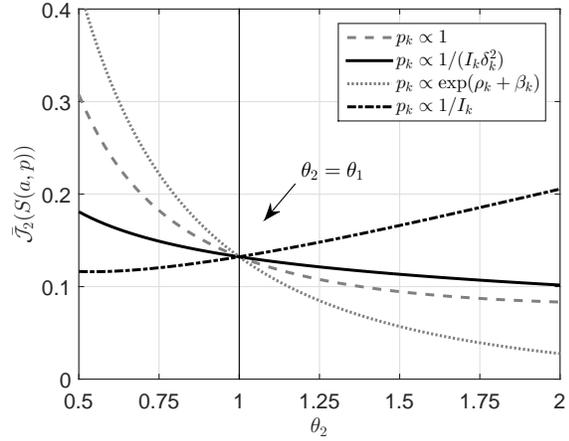}
        \caption{The change occurs in the \textit{second} sensor.}
       \nolabel{fig:theta_1_0_reloss:two}
    \end{subfigure}%
    \caption{The horizontal axis in both graphs represents $\theta_{2}$, the post-change mean in the second sensor when it is the one affected by the change. The post-change mean in the first sensor, $\theta_1$,
  when this is the one affected by the change,  is fixed to 1. The vertical axis in (a)  (resp. (b)) represents the asymptotic  relative performance loss of the multichart CUSUM when the first (resp. second) sensor is affected.}
    \label{fig:theta_1_0_reloss}
\end{figure}

\section{A simulation study} \label{sec:simulation}
\subsection{Description}

In this section  we present the results of a simulation study whose main goal is to compare the  GLR-CUSUM, $S$, given by  \eqref{eq:glr_cusum:cusum}, the
mixture-based CUSUM rules, $\tilde{S}$ and $\hat{S}$, given by \eqref{eq:mixture_cusum:cusum} and \eqref{eq:mixture_cusum:bayes}, respectively,  and the SUM-CUSUM, $M$, defined in \eqref{eq:mixture_cusum:mei}. We set $K = 5$ and $h_{k} = \cN (0, 1)$, $g_{k} = \cN (1, 1)$. That is, all sensors initially observe iid Gaussian observations with variance $1$, and at the time of the change the mean changes from $0$ to $1$ in an unknown subset of these sensors.  For our comparisons to be fair, we need to guarantee that all detection rules have access to the same amount of prior information. We consider the following regimes:%
\begin{inparaenum}[(i)]
    \item no prior information ($\cP = \overline{\cP}_{K})$;
    \item knowing that \emph{at most} $L$ sensors can be affected  ($\cP = \overline{\cP}_{L}$),%
    \ignore{\item knowing that \textit{exactly} $L$ can be affected  ($\cP = \cP_{L}$).}
\end{inparaenum}
where $1 \leq L \leq K$.

For the implementation of $S$ and $\tilde{S}$ we follow the \emph{first} approach described in Subsection \ref{implement}. For the implementation of $\hat{S}$ we use the \emph{second} approach, together with the   adaptive-window based on regeneration times  \eqref{eq:window:mei}.

 For all  detection rules we have considered in this work, the worst case scenario  is when the change occurs at time $0$. Thus,
when the actual affected subset is  $\cA $, in  order to compute the worst-case detection delay of each rule we simply need to simulate it  under $\Pro _0 ^{\cA }$.

In order to have a direct comparison of the various detection rules and also to illustrate the  second-order asymptotic optimality property, we plot the \emph{additional} number of observations  required by each rule in order to detect the change relative to the  optimal CUSUM test for which  the  affected subset is known in advance. Thus, if $T$ represents a generic detection rule and $\cA$ the true subset of affected sensors, we plot $\Exp_{0}^{\cA}[ T] - \Exp_{0}^{\cA}[S^{\cA}]$ against $\text{ARL}[T] := \Exp_\infty[T]$ for various threshold choices. In order to illustrate the first-order asymptotic optimality property, we plot the ratio $\Exp_{0}^{\cA}[ T]  /  \Exp_{0}^{\cA}[S^{\cA}]$ against $\text{ARL}[T] = \Exp_\infty[T]$ again for various threshold choices.   In Table \ref{table:performance:theta_1_00}, we present numerical  results for all detection rules when we choose their thresholds so that their target for the expected time to false alarm is $\gamma = 10^5$.
Standard errors are given based on $50,000$ Monte Carlo simulation runs.



\subsection{Results}

\subsubsection{No prior information}

We present the results that correspond to the case of  no prior information in Fig.~\ref{fig:no_prior:ratio}--Fig.~\ref{fig:no_prior:diff},
illustrating the notion of first-order and second-order  asymptotic optimality with respect to class \ $\overline{\cP }_K$, respectively.
We see that when $2$ sensors are affected, the mixture-based CUSUM procedures  are  slightly worse than the GLR-CUSUM,  whereas when $4$ sensors are affected, the mixture CUSUMs take the lead. In both cases the proposed schemes perform much better than SUM-CUSUM, $M(5)$, whose inflicted performance loss (relative to the optimal performance) increases much faster as $\gamma $ increases. Nevertheless, its ratio over the optimal performance decreases, which supports the result that SUM-CUSUM is \textit{first}-order asymptotically optimal.

\begin{figure}[!htb]
    \centering
    \begin{subfigure}[b]{.5\linewidth}
        \centering
        \includegraphics[width=1.0\textwidth]{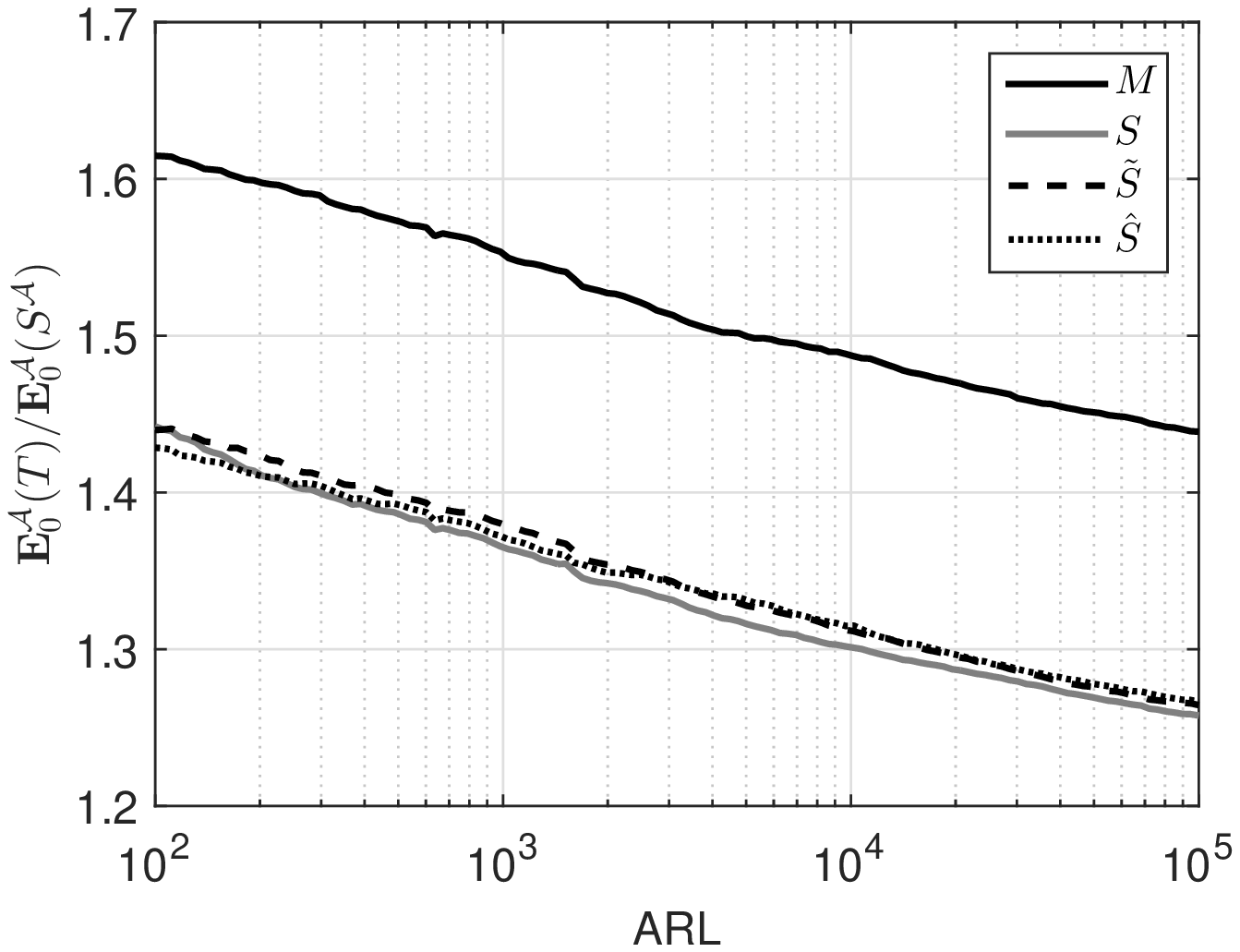}
        \caption{$2$ affected sensors.}
        \nolabel{fig:no_prior:ratio:2_affected}
    \end{subfigure}%
%
    \begin{subfigure}[b]{.5\linewidth}
        \centering
        \includegraphics[width=1.0\textwidth]{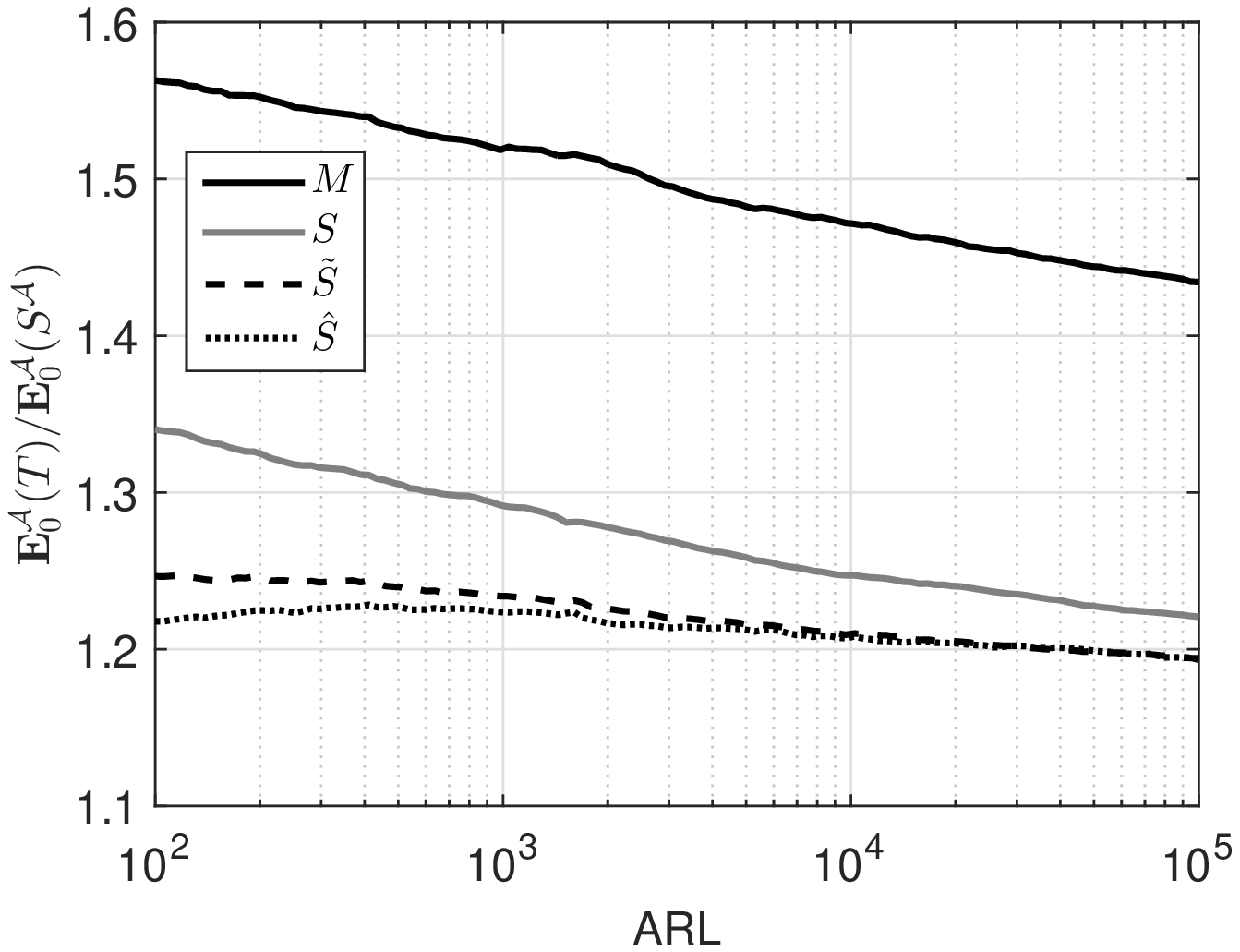}
        \caption{$4$  affected sensors.}
        \nolabel{fig:no_prior:ratio:4_affected}
    \end{subfigure}%
    \caption{ First-order asymptotic optimality is illustrated in the case of no prior information. The  performance ratio (vertical axis) of each rule is plotted against the expected time for a false alarm (horizontal axis).}

    \label{fig:no_prior:ratio}
\end{figure}
\begin{figure}[!htb]
    \centering
    \begin{subfigure}[b]{.5\linewidth}
        \centering
        \includegraphics[width=1.0\textwidth]{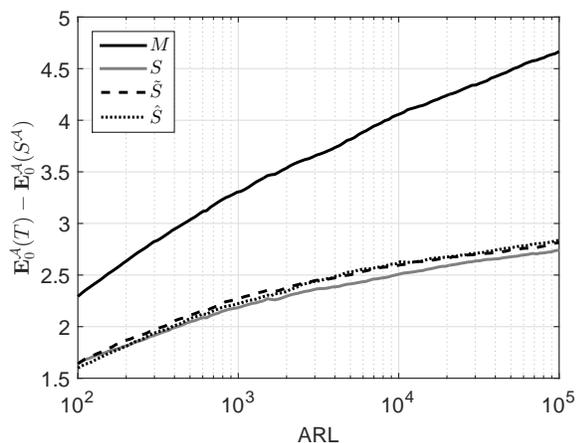}
        \caption{$2$ affected sensors.}
        \nolabel{fig:no_prior:diff:2_affected}
    \end{subfigure}%
%
    \begin{subfigure}[b]{.5\linewidth}
        \centering
        \includegraphics[width=1.0\textwidth]{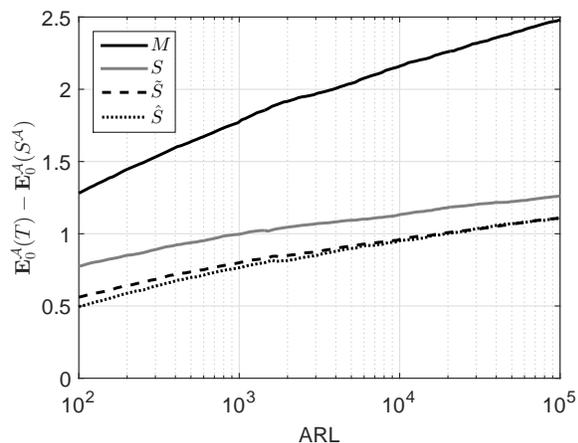}
        \caption{$4$  affected sensors.}
        \nolabel{fig:no_prior:diff:4_affected}
    \end{subfigure}%
    \caption{ Second-order asymptotic optimality is illustrated in the case of no prior information. The  performance loss (vertical axis) of each rule is plotted against the expected time for a false alarm (horizontal axis).}
    \label{fig:no_prior:diff}
\end{figure}

\subsubsection{An upper bound on the number of affected sensors is known}

In Fig.~\ref{fig:known_max:diff:max_3_affected}--\ref{fig:known_max:diff:max_4_affected}
we assume that  we know in advance that  \emph{at most} $L$ sensors can be affected and we compare  the GLR-CUSUM procedure, $S$, the mixture-based CUSUM procedures, $\tilde{S}$ and $\hat{S} (L / (2 K))$, and the scalable scheme $M(L)$. Again,  we see that $M(L)$ performs uniformly worse than the proposed schemes in all cases.

%
\begin{figure}[!htb]
    \centering
    \begin{subfigure}[b]{.5\linewidth}
        \centering
        \includegraphics[width=1.0\textwidth]{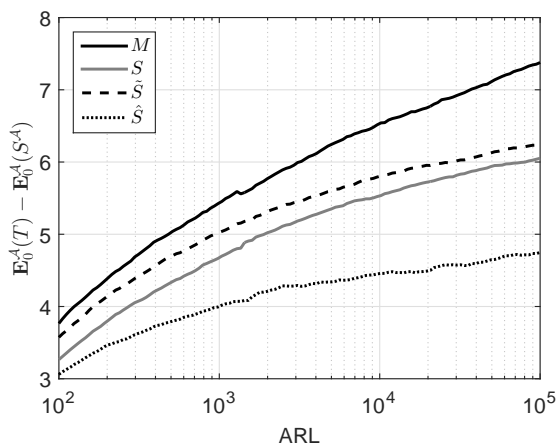}
        \caption{$1$ sensor actually affected }
        \nolabel{fig:known_max:diff:1_max_3_affected}
    \end{subfigure}%
%
    \begin{subfigure}[b]{.5\linewidth}
        \centering
        \includegraphics[width=1.0\textwidth]{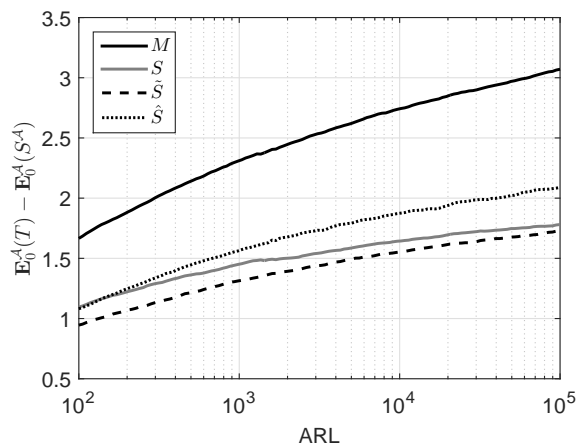}
        \caption{$3$  sensors actually affected}
        \nolabel{fig:known_max:diff:3_max_3_affected}
    \end{subfigure}%
    \caption{ Second-order asymptotic optimality  illustrated when it is known that \textit{ at most $3$ sensors may be affected}.  The  performance loss (vertical axis) of each rule is plotted against the expected time for a false alarm (horizontal axis).}
    \label{fig:known_max:diff:max_3_affected}
\end{figure}
\begin{figure}[!htb]
    \centering
    \begin{subfigure}[b]{.5\linewidth}
        \centering
        \includegraphics[width=1.0\textwidth]{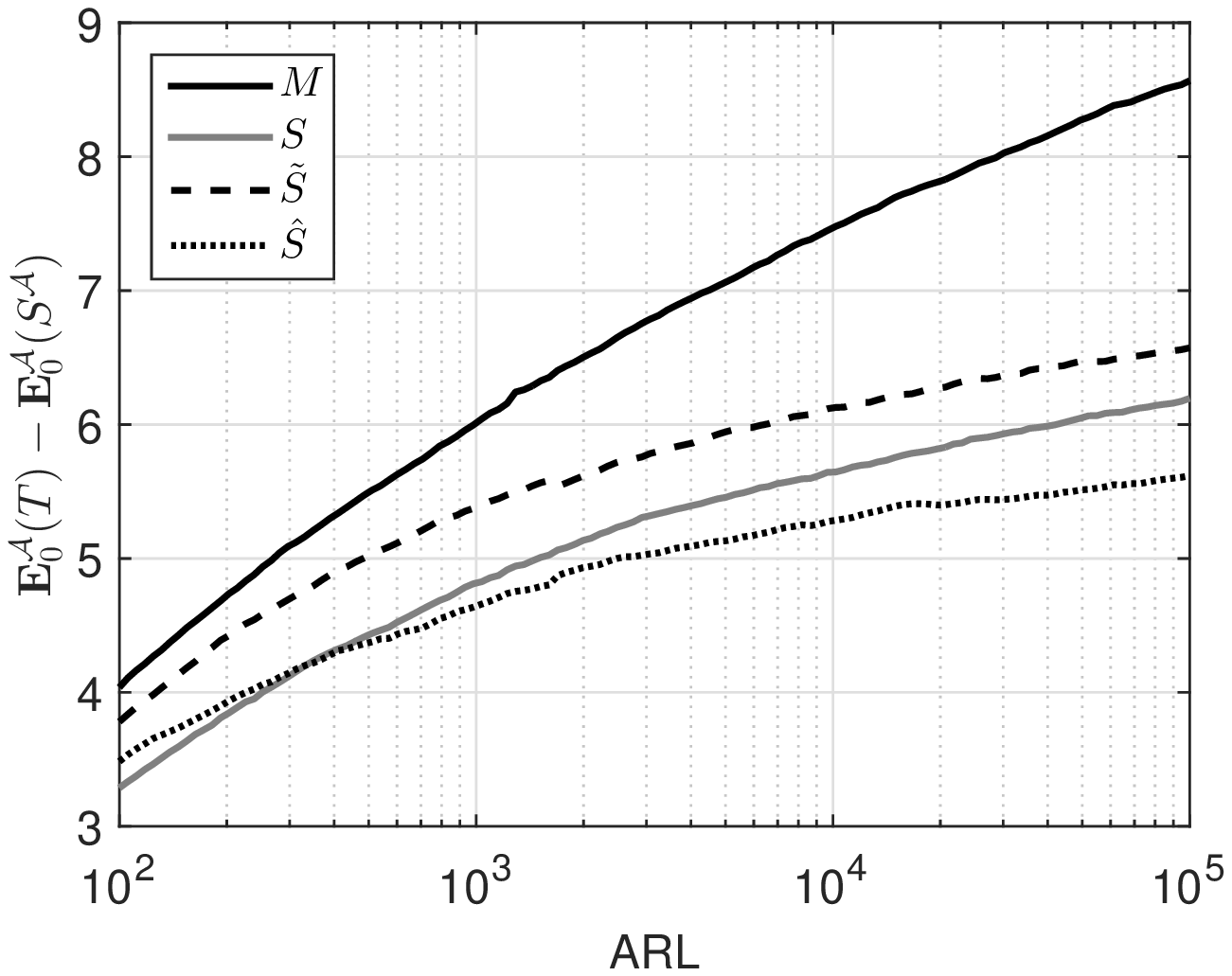}
        \caption{$1$  sensor actually affected}
        \nolabel{fig:known_max:diff:1_max_4_affected}
    \end{subfigure}%
%
    \begin{subfigure}[b]{.5\linewidth}
        \centering
        \includegraphics[width=1.0\textwidth]{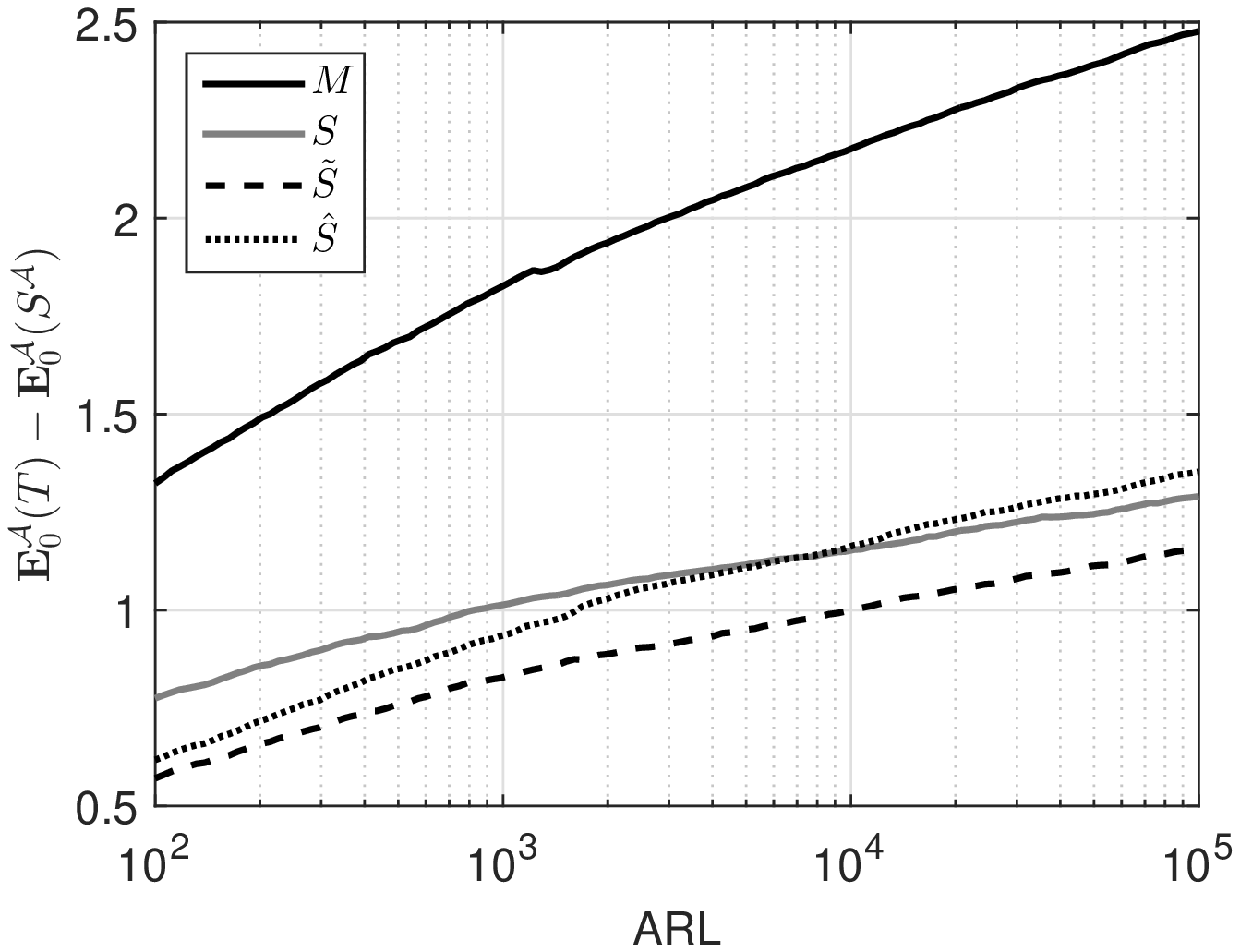}
        \caption{$4$ sensors actually affected}
        \nolabel{fig:known_max:diff:4_max_4_affected}
    \end{subfigure}%
    \caption{Second-order asymptotic optimality  illustrated when it is known that \textit{ at most $4$ sensors may be affected}.  The  performance loss (vertical axis) of each rule is plotted against the expected time for a false alarm (horizontal axis)}
    \label{fig:known_max:diff:max_4_affected}
\end{figure}

\begin{table}[!htb]
    \centering
    \caption{Performance of various procedures when $|\cA |$ sensors are affected. The thresholds are selected so that the expected time to false alarm is approximately $\gamma = 10^5$ .} \label{table:performance:theta_1_00}
    \hrule
    \begin{tabular}{l|c|c|cc|cc}
        & \hspace{-1em} ~$|\cA |$~ \hspace{-1em}  & $b$   & $\EV _\infty (T)$ & $\SE $  & $\EV _0 ^{\cA} (T)$ & $\SE $  \\ \hline
        $S_b^{\cA }$                         & $2$  & $9.88$ & $100090$ & $450$ & $10.64$ & $0.02$  \\
        $M_b(5)$                             & $2$  & $17.1$ & $100010$ & $455$ & $15.30$ & $0.03$  \\
        $M_b(2)$                             & $2$  & $14.2$ & $100065$ & $450$ & $14.21$ & $0.03$  \\
        $\hat{S} _b (0.5)$                   & $2$  & $9.85$ & $100065$ & $450$ & $13.47$ & $0.03$  \\
        $\hat{S} _b (0.2)$                   & $2$  & $9.35$ & $100105$ & $450$ & $13.57$ & $0.03$  \\
        $\tilde{S}_b(\overline{\cP } _{5})$  & $2$  & $9.91$ & $100105$ & $450$ & $13.45$ & $0.03$  \\
        $\tilde{S}_b(\overline{\cP } _{2})$  & $2$  & $9.86$ & $100115$ & $465$ & $13.12$ & $0.03$  \\
        $S_b(\overline{\cP } _{5})$          & $2$  & $9.58$ & $100005$ & $445$ & $13.38$ & $0.03$  \\
        $S_b(\overline{\cP } _{2})$          & $2$  & $9.78$ & $100005$ & $460$ & $13.15$ & $0.03$  \\ \hline

        $S_b^{\cA }$                         & $3$  & $9.94$ & $100065$ & $450$ & $7.369$ & $0.02$  \\
        $M_b(5)$                             & $3$  & $17.1$ & $100010$ & $455$ & $10.59$ & $0.02$  \\
        $M_b(3)$                             & $3$  & $15.9$ & $100025$ & $450$ & $10.44$ & $0.02$  \\
        $\hat{S} _b (0.5)$                   & $3$  & $9.85$ & $100065$ & $450$ & $9.040$ & $0.02$  \\
        $\hat{S} _b (0.3)$                   & $3$  & $9.63$ & $100075$ & $460$ & $9.458$ & $0.02$  \\
        $\tilde{S}_b(\overline{\cP } _{5})$  & $3$  & $9.91$ & $100105$ & $450$ & $9.054$ & $0.02$  \\
        $\tilde{S}_b(\overline{\cP } _{3})$  & $3$  & $9.90$ & $100115$ & $450$ & $9.098$ & $0.02$  \\
        $S_b(\overline{\cP } _{5})$          & $3$  & $9.58$ & $100005$ & $445$ & $9.136$ & $0.02$  \\
        $S_b(\overline{\cP } _{3})$          & $3$  & $9.67$ & $100060$ & $450$ & $9.150$ & $0.02$  \\ \hline

        $S_b^{\cA }$                         & $4$  & $9.93$ & $100010$ & $450$ & $5.716$ & $0.02$  \\
        $M_b(5)$                             & $4$  & $17.1$ & $100010$ & $455$ & $8.197$ & $0.02$  \\
        $M_b(4)$                             & $4$  & $16.8$ & $100050$ & $450$ & $8.192$ & $0.02$  \\
        $\hat{S} _b (0.5)$                   & $4$  & $9.85$ & $100065$ & $450$ & $6.821$ & $0.02$  \\
        $\hat{S} _b (0.4)$                   & $4$  & $9.75$ & $100050$ & $450$ & $7.068$ & $0.02$  \\
        $\tilde{S}_b(\overline{\cP } _{5})$  & $4$  & $9.91$ & $100105$ & $450$ & $6.826$ & $0.02$  \\
        $\tilde{S}_b(\overline{\cP } _{4})$  & $4$  & $9.91$ & $100070$ & $450$ & $6.870$ & $0.02$  \\
        $S_b(\overline{\cP } _{5})$          & $4$  & $9.58$ & $100005$ & $445$ & $6.977$ & $0.02$  \\
        $S_b(\overline{\cP } _{4})$          & $4$  & $9.60$ & $100010$ & $450$ & $7.006$ & $0.02$  \\ \hline
    \end{tabular}
    \hrule
\end{table}

\section{Conclusions} \label{sec:conclusions}
 We considered a generalized multisensor sequential change  detection problem, in which a number of possibly correlated sensors are monitored on line,  a global parameter vector determines their joint distribution and there is a change in an unknown subset of the components of this parameter vector. In this context, we established a strong asymptotic optimality property for various CUSUM-based detection rules.  In the special case that the sensors are independent and only a subset of them is affected by the change,  we proposed feasible versions of the above procedures.  We showed using simulation experiments that the proposed detection rules always outperform  scalable detection rules, such as the one proposed in \cite{mei_bio}, and we argued that this   phenomenon can be explained theoretically by the fact that the latter scheme enjoys a weaker form of asymptotic optimality.  Moreover, we proposed a modification of a detection rule in \cite{xie} that is  able to incorporate prior information.  Finally, we  considered the design of the multichart CUSUM in the special case that  the change is known to affect  exactly one sensor.

\section{Acknowledgments}
We would like to thank Drs. Alexander Tartakovsky, George Moustakides, Venu Veeravalli and Emmanuel Yashchin for  stimulating discussions and helpful suggestions, as well as the two referees whose comments helped us improve significantly an earlier version of this paper. The work of the first author was supported by the US National Science Foundation under Grant CCF 1514245, as well as by a collaboration grant  from  the   Simons Foundation.

\appendix
\section*{Appendix}\label{appen}

\begin{proof} [Proof of Theorem \ref{thm:second-order:glr_mix_cusum}]
    From the definition of the three detection rules it is clear that  $\tilde{S}_b \leq S_b$ and $\bar{S}_b \leq S_b$, therefore for any $b$ and $\cA$ we have:
    \[
        \cJ_\cA[\bar{S}_b] \leq \cJ_\cA[S_b] \quad \text{and} \quad   \Exp_\infty[S_b] \geq \Exp_\infty[\bar{S}_b]
    \]
    and
    \[
        \cJ_\cA[\tilde{S}_b] \leq \cJ_\cA[S_b] \quad \text{and} \quad   \Exp_\infty[S_b] \geq \Exp_\infty[\tilde{S}_b].
    \]
    Consequently,  in order to prove (i), it suffices to show that for any $b>0$
    \begin{align} \label{eq:lower_bounds}
       \Exp_{\infty}[\bar{S}_b]  \geq  e^{b} \quad \text{and}
        \quad \Exp_{\infty}[\tilde{S}_b]  \geq  e^{b} ,
    \end{align}
and in order to establish (ii), it suffices to show that as $b \rightarrow \infty$
    \begin{align} \label{eq:upper_bounds}
     \cJ_{\cA} [S_b] \leq \frac{b}{ I_{\cA}}  + \calo(1).
    \end{align}

    (i) In order to prove the first inequality in \eqref{eq:lower_bounds}, we introduce the mixture probability measure
    \[
        \bar{\Pro} := \sum_{\cA \in \cP} p_{\cA} \Pro_{0}^{\cA},
    \]
    and denote by $\overline{\Exp}$ the corresponding expectation. The mixture-based CUSUM rule  $\bar{S}_b$ results from repeated application, in the spirit of Lorden's \cite{lorden1} construction,  of the one-sided sequential test:
    \begin{align*}
        \bar{T}_{b}  &:= \inf \left\{ t \geq 0 : \bar{\Lambda}_{t}  \geq e^{b} \right\},
    \end{align*}
    where $\bar{\Lambda}$ is the likelihood ratio process of $\bar{\Pro}$ versus $\Pro_\infty$, i.e.,
    \[
        \bar{\Lambda}_{t}:= \frac{d\bar{\Pro}}{d\Pro_\infty}(\cFt) = \sum_{\cA \in \cP} p_{\cA} \Lambda_{t}^{\cA}, \quad t \in \bN.
    \]
    Consequently,   from   \cite[Th.\ 2]{lorden1} it follows that for every $b>0$ we have
    \[
        \Exp_\infty[\bar{S}_b]\geq \frac{1}{\Pro_\infty( \bar{T}_b <\infty)},
    \]
    and from Wald's likelihood ratio identity we obtain
    \begin{align*}
        \Pro_\infty( \bar{T}_{b} <\infty)
        &=  \overline{\Exp} \left[ \bar{\Lambda}^{-1}_{T_b} \,  \One_{ \{ \bar{T}_{b} < \infty \} } \right] \leq e^{-b},
    \end{align*}
    which completes the proof of the first inequality in \eqref{eq:lower_bounds}. In order to establish the second inequality, we introduce for each subset $\cA \in \cP$ the corresponding Shiryaev-Roberts statistic
    \[
        R_{t}^{\cA}:=\sum_{s=0}^{t-1} \exp(Z_t^{\cA}-Z_s^{\cA}), \quad R_0^{\cA}:=0
    \]
    and observe that
    \[
        R_{t}^{\cA} \geq \max_{0 \leq s <  t}  \exp(Z_t^{\cA}-Z_s^{\cA}) = \exp( \tilde{Y}_{t}^{\cA}) .
    \]
We  also introduce the  mixture Shiryaev-Roberts statistic
   $
       R_{t} := \sum_{\cA \in \cP} p_{\cA} R_{t}^{\cA}
  $
    and the corresponding stopping time
    $
       \cR_b := \inf\left\{t \geq 0:  R_{t} \geq e^{b} \right\}.
    $
    Then,  for any $b>0$ we obtain
    \begin{align*}
        \tilde{S}_b &= \inf\left\{t \in \bN : \sum_{\cA \in \cP} p_{\cA} \exp( \tilde{Y}_{t}^{\cA}) \geq e^{b} \right\} \\
               &\geq \inf\left\{t \in \bN: \sum_{\cA \in \cP} p_{\cA} R_t^{\cA} \geq e^{b} \right\}                = \cR_b,
    \end{align*}
thus,  it suffices to show that $\Exp_{\infty}[\cR_b]  \geq  e^{b}$. In order to see this, note that  for any $b$ and $t$ we have   $0 \leq R_{t} < e^{b}$ on the event  $\{ \cR_b > t \}$, therefore
  \begin{align*}
      \int _{\{ \cR _b > t \}} | R_{t} - t| \, d \Pro _{\infty }
        &\leq \int _{\{ \cR _b > t \}}  R_{t}   \, d \Pro _{\infty }
        +   \int _{\{ \cR _b > t \}}  t   \, d \Pro _{\infty } \\
         &\leq   e^{b} \, \Pro_\infty( \cR _b > t  )    +   \int _{\{ \cR _b > t \}}  \cR_b  \, d \Pro _{\infty }.
   \end{align*}
For any given $b$, the upper bound goes to 0  as $t \rightarrow \infty$ as long as  $\Exp_{\infty}[\cR_b]  <\infty $. But this can be assumed  without any  loss of generality, since otherwise the desired   inequality holds  trivially. Therefore, we can apply the optional sampling theorem to the
$\Pro_\infty$-martingale $\{R_{t}-t\}_{t \in \bN}$  and the stopping time $\cR_b$ and obtain
   $
        \Exp_{\infty}[ \cR_b]= \Exp_\infty[ R_{\cR_b}] \geq e^{b},
   $
 where the  inequality holds from the definition of the stopping rule $\cR_b$. This completes the proof of the second inequality in \eqref{eq:lower_bounds}.

    (ii) For any $b>0$ and $\cA \in \cP$,  from the definition of $S_b$ in \eqref{eq:glr_cusum:cusum} it follows that  $S_b \leq S_{b-\log p_{\cA}}^{\cA}$, where $S_{b-\log p_{\cA}}^{\cA}$ is the CUSUM rule   defined in \eqref{eq:cusum} with threshold $b-\log p_\cA$. Consequently, as $b \rightarrow \infty$ we have
    \begin{align*}
    \cJ_{\cA} [S_b] &\leq
              \frac{b- \log p_{\cA} }{ I_{\cA}}  + \calo(1)
           = \frac{b }{ I_{\cA}}  + \calo(1),
    \end{align*}
which completes the proof of     \eqref{eq:upper_bounds}.

\end{proof}

\begin{proof}[Proof of Proposition \ref{prop1}]
Fix $b>0$ and $\cA \in \cP$.   (i)  We will first prove the result for the GLR-CUSUM, $S_b$.  On the event $\{S_b>\nu\}$ we have
    \[
        S_b  =\inf \left\{t >\nu:    \max_{\cB \in \cP}  \left(Y_{t}^{\cB} -Y_{\nu}^{\cB}+
        Y_{\nu}^{\cB} +\log p_{\cB} \right) \geq b \right\} .
    \]
 For each $\cB \in \cP$,   $\{Y_{t}^{\cB} -Y_{\nu}^{\cB}\}_{t > \nu}$ depends only on the post-change observations. Therefore,  the detection statistic depends on the pre-change observations only through the non-negative random variables  $Y_{\nu}^{\cB}$, $\cB \in \cP$. Moreover, the  detection statistic is increasing in these quantities, which means that, for any possible change point $\nu$,  the worst-possible scenario for the history of observations up to $\nu$ is that $Y_{\nu}^{\cB}=0$ for every $\cB \in \cP$:
     \begin{align*}
        \esssup \Exp^{\cA}_{\nu}[ \left(S_b- \nu \right)^{+} | \cF_{\nu}]
            &= \Exp^{\cA}_{\nu}[ \left(S_b- \nu \right)^{+}  | \, Y^{\cB}_{\nu}=0, \,  \forall \, \cB \in \cP].
    \end{align*}
  However, when  $Y_{\nu}^{\cB}=0$ for every $\cB \in \cP$, each  sequence  $\{Y_t^{\cB}-Y_{\nu}^{\cB} \}_{t \geq \nu}$ under $\Pro_\nu^\cA$ has the same distribution as $\{Y^{\cB}_t\}_{t \geq 0}$ under $\Pro_0^\cA$.   Thus, for every change-point $\nu$ we have
    \begin{align*}
      \Exp^{\cA}_{\nu}[ \left(S_b- \nu \right)^{+}  | \, Y^{\cB}_{\nu}=0, \,  \forall \, \cB \in \cP]  &= \Exp_{0}^{\cA}[S_b],
    \end{align*}
which proves that     $\cJ_{\cA}[S_b] = \Exp_{0} ^{\cA } [S_b]$.

(ii) In order to prove the     corresponding result for $\bar{S}_b$, we note that  on the event $\{ \bar{S}_b > \nu\}$ we have
    \begin{align*}
            \bar{S}_b &= \inf \left\{ t >\nu : \max _{0 \leq s \leq t} \sum _{\cB \in \cP}  p_{\cB}  \exp(Z_t^\cB - Z_s^\cB) \geq  e^{b} \right\} \\
            & \leq  \inf \left\{ t >\nu : \max _{\nu \leq  s \leq t} \sum _{\cB \in \cP}  p_{\cB}  \exp(Z_t^\cB - Z_s^\cB) \geq  e^{b} \right\} \\
            &:= \bar{S}_b^{(\nu)} .
    \end{align*}
For every  change point $\nu$,   $\bar{S}_b^{(\nu)}$ under $\Pro_\nu^\cA$  has the same distribution as $\bar{S}_b$ under $\Pro_0^\cA$ and is independent of $\cF_\nu$. Therefore,
    \begin{align*}
       \Exp_{\nu }^{\cA } \left[  ( \bar{S}_b - \nu )^{+} | \cF _{\nu} \right]
            &  \leq \Exp _{\nu} ^{\cA } [\bar{S}_b^{(\nu)}]
                 = \Exp_{0} ^{\cA } [\bar{S}_b],
    \end{align*}
   which proves that  $\cJ_{\cA}[\bar{S}_b] \leq \Exp_{0} ^{\cA } [\bar{S}_b]$, and consequently  $\cJ_{\cA}[\bar{S}_b] = \Exp_{0} ^{\cA } [\bar{S}_b]$.

(iii) Finally, on the event $\{ \tilde{S}_b > \nu\}$ we have
     \begin{align*}
       \tilde{S}_b &=\inf \left\{t >\nu:    \sum_{\cB \in \cP}  p_{\cB} \, \exp\{ \tilde{Y}_{t}^{\cB}  \} \geq b \right\}  \\
        &=\inf \left\{t >\nu:    \sum_{\cB \in \cP}  p_{\cB} \, \exp\{ \tilde{Y}_{t}^{\cB} -\tilde{Y}_{\nu}^{\cB}+        \tilde{Y}_{\nu}^{\cB}  \} \geq b \right\} .
       \end{align*}
Fix $\cB \in \cP$.  When $\tilde{Y}_{\nu }^{\cB } < 0$,  then    from
 recursion \eqref{recursion2} it follows that   $\sset{\tilde{Y}_{t}^{\cB }}_{t > \nu }$  is independent of  $\cF_\nu$ and that its
distribution    under $\Pr _{\nu } ^{\cA }$
is the same as that  of $\sset{\tilde{Y}_t^{\cB }}_{t >0}$ under $\Pr _{0} ^{\cA }$.   When  $\tilde{Y}_{\nu }^{\cB } \geq 0$, a similar argument as in (i) shows that the worst case is $\tilde{Y}_{\nu }^{\cB } =0$. Thus, we conclude that
 $\cJ_{\cA}[\tilde{S}_b] = \EV _{0} ^{\cA } [\tilde{S}_b]$.

 \end{proof}

\begin{proof}[Proof of Proposition \ref{prop2}]
Recall the notation $C(\cP):=\sum_{\cB \in \cP} p^{|\cB|}$.

(i) We observe that  $\log  C(\cP_L) = L \log p +\log |\cP_L|$ and
  \begin{align*}
        \max_{\cA \in \cP_{L}}  \sum_{k \in \cA} \left(Z_{t}^{k}-Z_{s}^{k} +\log p \right) &=  \sum _{k = 1}^L  \left(Z_{s:t}^{(k)} + \log p \right),
    \end{align*}
which proves \eqref{eq:glr_cusum:impl_exact}. Moreover, we observe that
    $
         \log  C(\overline{\cP}_{L}) =
        \log \left( \sum_{k=1}^{L}  |\cP_k| \,  p^{k}  \right)
    $
    and that
    \[
        \max_{\cA \in \cP}  \sum_{k \in \cA} \left(Z_{t}^{k}-Z_{s}^{k} +\log p \right) = \sum _{k = 1}^L \left( Z_{s:t}^{(k)} +\log p \right)^{+}
    \]
    whenever there is  at least one  $k$ such that $Z_{t}^{k}-Z_s^k+\log p>0$.     On the other hand,  if $Z_{t}^{k}-Z_s^k+\log p \leq 0$ for every $1 \leq k \leq K$,  the  detection statistic cannot reach a  positive threshold, which
 completes the proof of \eqref{eq:glr_cusum:impl_at_most}.

(ii)  We observe that  for any $\pi \in (0,1)$ and $p=\pi /(1-\pi)$ we have
    \begin{align*}
&    \prod_{k=1}^{K} \left( 1-\pi +\pi \, \frac{\Lambda_t^k}{\Lambda_s^k} \right) \\
    &=  (1-\pi)^{K}+ \sum_{\cA \in \overline{\cP}_{K} } (1-\pi)^{K-|\cA|} \, \pi^{|\cA|} \prod_{k \in \cA}  \frac{\Lambda_t^k}{\Lambda_s^k}  \\
    &=  (1-\pi)^{K} \left[  1+ \sum_{\cA \in \overline{\cP}_{K} }  p^{|\cA|} \,  \frac{\Lambda_t^\cA}{\Lambda_s^\cA}  \right] ,
    \end{align*}
which proves  \eqref{eq:mixture_cusum:bayes}.

\end{proof}

\begin{proof}[Proof of Theorem  \ref{thm:second-order:bayes_sieg_cusum}]
    (i)   Fix $b$. We need to show that
    \begin{align} \label{eq:aux:bounds_on_arl}
        \Exp_\infty[ \hat{S}_{b}(\pi) ]  &\geq  e^{b} \quad \text{and} \quad
        \Exp_\infty[ \check{S}_{b}(\pi) ]  \geq  \frac{e^{b}}{|\overline{\cP}_{K}|},
    \end{align}
    where $\pi \in (0,1)$ in the first and $\pi \in (0,1]$ in the second inequality.
    In order to prove the first inequality, we recall that $\hat{S}_{b}(\pi)$ can be obtained with the repeated application of the one-sided sequential test $\hat{T}_{b}(\pi)$, defined in \eqref{eq:mixture_sprt:bayes}.
    Therefore,  from \cite[Th.\ 2]{lorden1} it follows that
    \[
        \Exp_\infty[\hat{S}_{b}(\pi)] \geq \frac{1}{\Pro_{\infty}(    \hat{T}_{b}(\pi) <\infty)}
    \]
    and from  Wald's likelihood ratio identity we obtain
    \begin{align*} \nolabel{eq:wald}
        \Pro_{\infty}(    \hat{T}_{b}(\pi) <\infty) &= 
            \hat{\Exp} \left[ \exp\left\{ -\hat{Z}_{  \hat{T}_{b}(\pi) } \right\}
            \; \One _{ \{  \hat{T}_{b}(\pi) < \infty \} } \right]    \leq  e^{-b},
    \end{align*}
    where $\hat{Z}$ is the logarithm of the likelihood ratio  $\hat{\Lambda}^{\pi}$ defined in \eqref{eq:mixture_sprt:bayes}, i.e.,
    \begin{align*}
        \hat{Z}_{t}  &:= \log   \hat{\Lambda}_{t}
            = \sum_{k=1}^{K}  \log  \left(1 - \pi + \pi \,  \Lambda_{t}^{k} \right).
    \end{align*}
    This completes the proof of the first inequality in \eqref{eq:aux:bounds_on_arl}. In order to prove the  second inequality, it suffices to observe that for any $b>0$ and $\pi \in (0,1]$  we have
    \[
        \hat{S}_b(\pi) \geq \check{S}_{b}(\pi) \geq  \inf \Bigl\{ t \geq 0 : \max _{0 \leq s \leq t}
        \sum_{k=1}^{K}  (Z_{t}^{k}-Z_s^k)^{+}\geq b \Bigr\}
    \]
    and that the lower bound coincides (up to a constant) with the  unweighted GLR-CUSUM  in the case of complete uncertainty, that is the detection rule  obtained by setting  $L=K$ and $\pi=1$ in \eqref{eq:glr_cusum:impl_at_most}.


    (ii)  When $\pi=1$, $\check{S}_{b}(\pi)$ coincides with the GLR-CUSUM in the case of complete uncertainty, whose asymptotic optimality has already been established. Therefore, without loss of generality,  we assume that $\pi \in (0,1)$.  Due to the pathwise inequality
    $ \hat{S}_b(\pi) \geq \check{S}_{b}(\pi) $,
    it suffices to show that  for every subset $\cA$ we have
    \begin{align*} \nolabel{eq:mix-cusum-special:upper_bound}
        \cJ_{\cA}[\hat{S}_{b}(\pi)] &\leq  \frac{b}{I_{\cA}} + \calo(1),
    \end{align*}
    where $\calo(1)$ is a bounded term as $b \rightarrow \infty$.
    From \cite[Th.\ 2]{lorden1} it follows that
    $
        \cJ_{\cA}[   \hat{S}_{b}(\pi)]
            \leq  \Exp_{0}^{\cA}[   \hat{T}_{b}(\pi)]
    $
     for every $b>0$,      whereas  for every $t$ we observe that
    \begin{align*}
        \hat{Z}_{t} &= \sum_{k=1}^{K}  \log  \left(1 - \pi + \pi \,  \exp(Z_{t}^{k} ) \right) \\
            &\geq \sum_{k \in \cA}  (Z_{t}^{k} + \log \pi)+  \sum_{k \notin \cA}   \log (1-\pi) \\
            &=  Z_{t}^{\cA} + c_{\cA},
    \end{align*}
    where $c_{\cA} :=  |\cA|  \, \log \pi +  (K-|\cA|)  \, \log(1-\pi)$. Therefore,
   $
        \hat{T}_{b}(\pi)\leq  \inf\{t : Z_{t}^{\cA}  \geq  b+ c_{\cA} \}
   $
   for every $b>0$,    and consequently
    \begin{align*}
        \Exp_{0}^{\cA} [  \hat{T}_{b}(\pi) ] &\leq
          \frac{b+ c_{\cA}  }{I_{\cA}} +\calo(1),
    \end{align*}
    which completes the proof.
\end{proof}

\bibliographystyle{IEEEtran}
\bibliography{multisensor}

\end{document}